\documentclass[12 pt]{extarticle}
\usepackage{amsmath}
\usepackage{amssymb}
\usepackage[a4paper, total={6in, 9in}]{geometry}
\usepackage{amsthm,url,tikz}
\usepackage[utf8]{inputenc}
\usepackage[english]{babel}
\usepackage{graphicx}
\usepackage{hyperref}
\usepackage{adjustbox}
\numberwithin{equation}{section}
\usepackage{makecell}
\usepackage{bm}
\usepackage{xcolor}

\usepackage{enumerate}
\usepackage[nottoc,notlot,notlof]{tocbibind} 
\theoremstyle{plain} 
\newtheorem{lemma}{Lemma}[section] 
\newtheorem{theorem}[lemma]{Theorem}
\newtheorem{corollary}[lemma]{Corollary}
\newtheorem{proposition}[lemma]{Proposition}

\theoremstyle{definition} 
\newtheorem{definition}{Definition}[section] 
\newtheorem{remark}{Remark}[section] 
\newtheorem{example}{Example}

\title{Two Families of Linear Codes Containing Non-GRS MDS Codes}

\author{Kanat Abdukhalikov \\	
Department of Mathematical Sciences, \\
UAE University, PO Box 15551, Al Ain, UAE\\
Email: abdukhalik@uaeu.ac.ae \bigskip  \\  
Gyanendra K. Verma \\	
Department of Mathematical Sciences, \\
UAE University, PO Box 15551, Al Ain, UAE\\
and \\
Department of Electrical Communication Engineering, \\
Indian Institute of Science Bangalore, 560012, India\\
Email:  gkvermaiitdmaths@gmail.com}

\date{}

\begin{document}
	\maketitle
\begin{abstract} 
%Maximum distance separable (MDS) codes are optimal linear codes, and are in one-to-one correspondence with arcs in finite projective geometry. In this work, 
We construct two new families of linear codes by modifying the generator matrices of generalized Reed-Solomon (GRS) codes. For these codes, we explicitly derive parity-check matrices and establish necessary and sufficient conditions ensuring the MDS property. Additionally, we explore subfamilies within these constructions that are non-GRS MDS codes.  We also characterize their self-orthogonal and self-dual properties and present some explicit constructions and examples.
\end{abstract}

\section{Introduction}
Let $\mathbb{F}_q$ be a finite field with $q$ elements, where $q$ is a prime power. A linear code $C$ over $\mathbb{F}_q$ with parameters $[n,k,d]$ is a linear subspace of $\mathbb{F}_q^n$ with dimension $k$ and minimum Hamming distance $d$. The parameters of a linear code $C$ are constrained by the Singleton bound $d\leq n-k+1$. Codes attaining the bound, that is, $d=n-k+1$, are called maximum distance separable (MDS) codes. The study of MDS codes has been driven by their diverse applications, ranging from distributed storage systems and combinatorial designs to finite geometry, random error channels, and secret sharing schemes (see \cite{Cadambe2011,Hirschfeld1998,Macwilliams1977,Sakakibara2013}). Due to these applications, the properties of MDS codes have been extensively investigated, including their existence, classification, and balanced variants \cite{Dau2013,Dau2014,Kokkala2014}.

A linear MDS code is equivalent to an arc in finite projective geometry. Namely, an $[n,k]$ linear code $C$ over $\mathbb{F}_q$ generated by matrix $G$ is an MDS code if and only if the set consisting of columns of $G$ as coordinate vectors is an arc of size $n$ in projective space $PG(k-1,q)$ (see \cite{Ball2019}).  Segre \cite{Segre1955} proposed an upper bound on the maximum size of an arc in $PG(k-1,q)$, known as the MDS conjecture, which states that if $3 \leq k \leq q-1$, then an arc in $PG(k-1,q)$
 has size at most $q+1$ except for the case when $q$ is even and $k=3$ or $k=q-1$ where arcs of size $q+2$ exist. Equivalently, in the language of coding theory, for an $[n,k,d]$ MDS code where $3\leq d< n$, the length $n$ satisfies $n \le q+1$, unless $k \in \{3, q-1\}$ and $q$ is even, in which case $n \le q+2$. The conjecture has been proved for several particular cases \cite{Ball2019}; for instance, Ball \cite{Ball2012} proved that the conjecture is true over prime fields. The general case of the conjecture remains open.  A prominent and well-studied family of MDS codes is the class of generalized Reed-Solomon (GRS) codes \cite{Macwilliams1977}. GRS codes over $\mathbb{F}_q$ can have length at most $q+1$. The GRS codes of length $q+1$ correspond to normal rational curves in projective spaces. All known $(q+1)$-arcs in $PG(k-1,q)$, $k\ge 4$, are normal rational curves (equivalently, GRS codes) except for the cases $k=5$, $q=9$ (Glynn \cite{Glynn1986}) and $k=4$, $q=2^m$ (Segre arcs \cite{Casse1982}). This leads to a conjecture proposed in \cite{Ball2019}, for $6\leq k\leq q-5$, all $q+1$ size arcs in $PG(k-1,q)$ are GRS codes. In view of the above statements, the study of non-GRS MDS codes remains interesting both from the perspective of theoretical classification of arcs in finite projective spaces (equivalently, MDS codes) and for practical applications in cryptography. For example, non-GRS MDS codes have applications in providing resistance against Sidelnikov-Shestakov and Wieschebrink attacks in cryptographic systems, where GRS codes do not offer such cryptographic security \cite{beelen2018,Lavauzelle2019}. Formally, MDS codes that are not  (monomial) equivalent to GRS codes are called non-generalized Reed-Solomon (non-GRS) MDS codes.

Roth and Lempel \cite{Roth1989} were the first to study non-GRS MDS codes by augmenting the generator matrices of GRS codes with two additional columns. The study of non-GRS MDS codes later gained renewed interest through the work of Beelen et al. \cite{beelen2018,Beelen2022}, who introduced twisted Reed-Solomon (TRS) codes and demonstrated that certain families of these codes contain non-GRS MDS codes. Recently, in \cite{bhagat2025,Jin2024,Li2024,Liu2025,Liu2024,Wu2024,Zhi2025} and references therein, the authors investigated several families of non-GRS MDS codes using a variety of techniques. These studies have led to the introduction of several new classes of codes, including twisted Reed–Solomon (TRS) codes, generalized twisted Reed–Solomon (GTRS) codes, non-GTRS codes, column twisted and row-column twisted Reed-Solomon codes that contain non-GRS MDS codes. In \cite{Han2024,Li2024,Zhi2025}, the authors studied the codes obtained by deleting the penultimate row of the GRS code generator matrices and extending the resulting codes. These constructions produced families of non-GRS MDS or NMDS codes that cannot be derived from GTRS codes. Furthermore, in \cite{Abdukhalikov2025nonrs,Jin2024}, the authors investigated the codes constructed by removing an arbitrary row from GRS code generator matrices and their extensions, established necessary and sufficient conditions for these codes to be non-GRS MDS codes, and provided explicit constructions.

A linear code is called a self-orthogonal code if $C\subseteq C^{\perp}$  and self-dual if  $C=C^{\perp}$, where $C^{\perp}$ is the dual code of $C$. Self-dual codes have connections with combinatorics and lattice theory, along with applications in cryptography (see \cite{Conway1999,Cramer2005,Dougherty2008, Macwilliams1977}), while self-orthogonal codes are particularly significant in quantum coding theory (for instance, see \cite{Calderbank1995,Shor1995}). The study of self-orthogonal and self-dual codes has thus become important both for advancing theoretical understanding and for enabling practical applications in classical and quantum coding. Recently, considerable attention has been given to the construction of self-orthogonal and self-dual MDS and NMDS codes, resulting in new families of codes with desirable properties and enhanced applicability (see \cite{Han2024,Li2024} and references therein). These developments not only enrich the classification of MDS and NMDS codes but also provide valuable tools for secure communication, reliable data storage, and quantum information processing.

Building on these previous developments, in this paper, we propose two new families of linear codes that contain non-GRS MDS codes,  namely $C_{i,j}$, $1\leq i<j\leq k$, and $C_{h,k}$.  Our contributions are as follows. We determine the parameters of $C_{i,j}$ and $C_{h,k}$ and establish necessary and sufficient conditions for these codes to be MDS. We explicitly determine parity-check matrices for $C_{i,j}$ and $C_{h,k}$.  We also investigate their non-GRS properties, using the Schur square method. Furthermore, we provide several explicit examples and constructions to illustrate and validate our results.

The rest of the paper is organized as follows. In Section \ref{pre}, we recall basic definitions and introduce the notation used throughout the paper. Sections \ref{firstfam} and \ref{secondfam} are devoted to the construction of two families of linear codes, where we study their MDS and  non-GRS properties, and explicitly derive their parity-check matrices. The theory of sysmmetric functions \cite{Macdonald1995} is very useful in our investigations. Furthermore, we investigate the self-orthogonal and self-dual properties of these codes in respective sections. In Section \ref{conclusion}, we conclude the paper and discuss potential directions for future work.

\section{Preliminaries}\label{pre}
Throughout the paper, we denote $\mathbb{F}_q$, a finite field with $q$ elements, where $q$ is a prime power, and $\mathbb{F}_q^*=\mathbb{F}_q\setminus\{0\}$. A $k\times n$ matrix $G$ is called a generator matrix for a linear code $C$ with parameters $[n,k]$ if the rows of $G$ form a basis of $C$. Two codes $C_1,C_2\subseteq \mathbb{F}_q^n$ are said to be (monomial) equivalent if $C_1$ can be obtained from $C_2$ by permutation of coordinates and multiplying coordinates by a nonzero scalar. The dual of $C$ is defined as
$$C^{\perp}=\{y\in \mathbb{F}_q^n: \langle y,c\rangle=0, \forall c\in C\},$$
where $\langle y,c\rangle=\sum_{i=1}^ny_ic_i$ is the Euclidean inner product. The code $C$ is called self-orthogonal if $C\subseteq C^{\perp}$ and called self-dual if $C=C^{\perp}$. Moreover, if the code $C$ is self-dual, then $n=2k$. 
The Singleton bound states that $d \le n-k+1$, where $d$ is the Hamming distance of $C$, defined as
$$d(C)=\min_{0\neq c\in C}\{wt(c): wt(c)=\text{number of nonzero coordinates in } c\}.$$
The code $C$ is said to be an MDS code if $d=n-k+1$. 
 A well-known class of MDS codes is generalized Reed-Solomon (GRS) codes defined as follows.
\begin{definition}
For $n\leq q$, let $\Lambda=\{\alpha_1,\alpha_2,\dots, \alpha_n\}\subseteq \mathbb{F}_q$ and $v=(v_1,v_2,\dots,v_n)\in (\mathbb{F}_q^*)^n$. The generalized Reed-Solomon (GRS) code is defined as 
$$GRS(n,k,\Lambda,v)=\{(v_1f(\alpha_1),v_2f(\alpha_2),\dots,v_nf(\alpha_n)) \mid  f(x)\in \mathbb{F}_q[x], \deg(f(x))\le k-1\}.$$ 
\end{definition}
In particular, when $v=(1,1,\dots,1)$, the code is called a Reed-Solomon (RS) code denoted by $RS(n,k,\Lambda)$. Note that the code $GRS(n,k,\Lambda,v)$ and $RS(n,k,\Lambda)$ are equivalent.  A generator matrix for $GRS(n,k,\Lambda,v)$ code is 
$$G=\begin{bmatrix}
    v_1&v_2& \dots &v_n\\
    v_1\alpha_1& v_2\alpha_2& \dots&v_n\alpha_n\\
    v_1\alpha_1^2& v_2\alpha_2^2& \dots&v_n\alpha_n^2\\
    \vdots &\vdots &\dots&\vdots\\
    v_1\alpha_1^{k-1}& v_2\alpha_2^{k-1}& \dots&v_n\alpha_n^{k-1}
\end{bmatrix}_{k\times n}.$$

\begin{definition}\label{nongrs} 
A non-GRS MDS code is an MDS code that is not equivalent to a GRS code. 
\end{definition}
The Schur (Hadamard) product (for details, see \cite{horn2013}) of two vectors $x=(x_1,x_2,\dots,x_n)$ and $y=(y_1,y_2,\dots,y_n)$ in $\mathbb{F}_q^n$ is defined as 
$$x\star y=(x_1y_1,x_2y_2,\dots,x_ny_n).$$
Let $C_1, C_2$ be codes with parameters $[n,k]$. Then the Schur product $C_1\star C_2$ is defined as
$$C_1\star C_2=\text{span}_{\mathbb{F}_q}\{c_1\star c_2: c_1\in C_1, c_2\in C_2\}.$$

If $C_1=C_2$, then the Schur product is called the Schur square. The following lemma provides a characterization of GRS codes using the Schur square.
\begin{lemma}\label{lemmanongrs}
(1) 
%\cite{couvreur2014} 
If $C$ is an $[n,k]$ GRS code with $\dim(C)\leq n/2$, then $\dim(C\star C)=2\dim(C)-1$.\\
(2) \cite{Mirandola2015} Let $C$ be an MDS code of length $n$ with $\dim(C)\leq \frac{n-1}{2}$. Then $C$ is a GRS code if and only if $\dim(C\star C)=2\dim(C)-1$.   
\end{lemma}

%\begin{remark}
% Note that several authors say a code is a non-GRS code if the code is MDS and not monomial equivalent to a GRS code. However, a code which is monomial equivalent to a GRS code is itself another GRS code. It can be seen as follows:
% Let $C$ be an $[n,k]$ MDS code which is monomial equivalent to a GRS code $C_{GRS}$ given by 

%$$C_{GRS}=\{(v_1f(\alpha_1),v_2f(\alpha_2),\dots,v_nf(\alpha_n)) \mid  f(x)\in V_k\}$$
%for some  $\Lambda=\{\alpha_1,\alpha_2,\dots \alpha_n\}\subseteq \mathbb{F}_q\cup\{\infty\}$ and   $v=(v_1,v_2,\dots,v_n)\in \mathbb{F}_q^n$ with $v_i\in \mathbb{F}_q^*$ for all $1\leq i\leq n$. Then there exist a permutation $\pi$ on $\{1,2,\dots,n\}$ and nonzero $u_i$'s in $\mathbb{F}_q$ for $1\leq i\leq n$ such that for all codeword $c=(c_1,c_2,\dots,c_n)\in C$, we have $$c_i=u_iv_{\pi(i)}f(\alpha_{\pi(i)})$$
% for all $1\leq i\leq n$. Thus $C$ is a GRS code obtained by evaluating polynomials in $V_k$ at $\{\pi(\alpha_1),\pi(\alpha_2),\dots,\pi(\alpha_n)\}$ and associated scalars  $(u_1v_{\pi(1)},u_2v_{\pi(2)},\dots,u_nv_{\pi(n)})$.
%\end{remark}
The following lemma is a well known characterization of MDS codes in terms of generator matrices. 
\begin{lemma}\cite[p.319]{Macwilliams1977}\label{mdsnonsingular}
 Let $C$ be an $[n,k]$ code with a generator matrix $G$. Then $C$ is MDS if and only if all $k\times k$ sub-matrices of $G$ are non-singular.   
\end{lemma}

Let $\Lambda=\{\alpha_1,\alpha_2,\dots,\alpha_n\}\subseteq \mathbb{F}_q$. Let $M_{\Lambda}$ be a square Vandermonde matrix of order $n$ defined by 
$$M_{\Lambda}=\begin{bmatrix}
    1&1&\dots&1\\
    \alpha_1&\alpha_2&\dots&\alpha_n\\
    \vdots&\vdots&\vdots&\vdots\\
    \alpha_1^{n-2}&\alpha_2^{n-2}&\dots&\alpha_n^{n-2}\\
    \alpha_1^{n-1}&\alpha_2^{n-1}&\dots&\alpha_n^{n-1}
\end{bmatrix}.$$
It is well known that the determinant of $M_\Lambda$ is  $\det(M_{\Lambda})=\prod_{1\leq i< j\leq n}(\alpha_j-\alpha_i)=D(\alpha_1,\dots,\alpha_n)$. 
Now, we fix some notation that will be used throughout the paper. For an integer $t\geq 0$, the $t^{\rm{th}}$ elementary symmetric function in $n$ variables is defined as
$$\sigma_t(x_1,x_2,\dots, x_n)=\sum_{1\leq j_1<j_1<\cdots<j_t\leq n}x_{i_{j_1}}x_{i_{j_2}}\cdots x_{i_{j_t}}$$
with $\sigma_0(x_1,x_2,\dots,x_n)=1$ and $\sigma_t(x_1,x_2,\dots, x_n)=0$ for all $t>n$.

Let $S_t(x_1,x_2,\dots,x_n)$ be the complete symmetric polynomial of degree $t$ defined by 
$$S_t(x_1,x_2,\dots,x_n)=\sum_{t_1+t_2+\cdots+t_n=t,\ t_i\geq 0}x_1^{t_1}x_2^{t_2}\cdots x_n^{t_n}.$$ 
Let $\Lambda=\{\alpha_1,\alpha_2\dots,\alpha_n\}\subseteq \mathbb{F}_q$.  By Cramer's rule \cite{cramer1750} , the solution of the system of linear equations 
$$\begin{bmatrix}
    1&1& \dots &1\\
    \alpha_1& \alpha_2& \dots&\alpha_n\\
    \alpha_1^2& \alpha_2^2& \dots&\alpha_n^2\\
    \vdots &\vdots &\dots&\vdots\\
    
    \alpha_1^{n-2}& \alpha_2^{n-2}& \dots&\alpha_n^{n-2}\\
    \alpha_1^{n-1}& \alpha_2^{n-1}& \dots&\alpha_n^{n-1}\\
\end{bmatrix}\begin{bmatrix}
    x_1\\
    x_2\\
    x_3\\
    \vdots\\
    x_{n-1}\\
      x_{n}\\
\end{bmatrix}=\begin{bmatrix}
    0\\
    0\\
    0\\
    \vdots\\
    0 \\
    1
\end{bmatrix}$$
is given by 
\begin{equation}\label{ui}
u_i=\prod_{1\leq j\leq n, j\neq i}(\alpha_i-\alpha_j)^{-1}, \ 1\leq i\leq n.
\end{equation}
 
 We denote $S_t(\alpha_1,\alpha_2,\dots,\alpha_n)$ by $S_t$ with $S_0=1$ and $\sigma_t(\alpha_1,\alpha_2,\dots,\alpha_n)$ by $\sigma_t$. For uniformity, we write $S_t=0$ for $t<0$. There is a fundamental relation between the elementary symmetric functions and the complete symmetric functions (for instance, see \cite[Page 21, Eq. $2.6'$]{Macdonald1995})
\begin{equation}\label{eqA}
    \sum_{t=0}^N(-1)^{t}\sigma_tS_{N-t}=0, \ \ \  \text{for all } N\geq 1.
\end{equation}
% Also, for $1\leq i\leq n$ and $r\geq 1$, define 

% $$\Lambda_{r,i}=\sum_{j=0}^r(-1)^j\sigma_j\alpha_i^{(n-k)+(r-1)-j}.$$

% The following proposition determines the determinant of a matrix obtained by removing an arbitrary row from a Vandermonde matrix and adding another row with an exponent equal to the order of the matrix. We frequently utilize this result to study MDS codes presented in this paper.

% \begin{proposition}\cite[Sec. 337]{Muir1960}\label{detk-r}
% For $1\leq r\leq n-1$, let 
% \begin{equation*}
%    M_n^{(r)}=\begin{bmatrix}
%     1&1& \dots &1\\
%     \alpha_1& \alpha_2& \dots&\alpha_n\\
%     \alpha_1^2& \alpha_2^2& \dots&\alpha_n^2\\
%     \vdots &\vdots &\dots&\vdots\\
%     \alpha_1^{n-r-1}& \alpha_2^{n-r-1}& \dots&\alpha_n^{k-r-1}\\
%     \alpha_1^{n-r+1}& \alpha_2^{n-r+1}& \dots&\alpha_n^{n-r+1}\\
%     \vdots &\vdots &\dots&\vdots\\
%     \alpha_1^{n-1}& \alpha_2^{n-1}& \dots&\alpha_n^{n-1}\\
%     \alpha_1^{n}& \alpha_2^{n}& \dots&\alpha_n^{n}\\
% \end{bmatrix}_{n\times n}.
% \end{equation*}
% Then \begin{equation*}
%     \det(M_n^{(r)})= \det(M_{\Lambda})\cdot \sigma_r(\alpha_1,\alpha_2,\dots,\alpha_n).
%     \end{equation*}
% \end{proposition}

Let $n$ and $h$ be positive integers. Define a generalized Vandermonde matrix $G_h$ as 
$$G_h=\begin{bmatrix}
    1&1&\dots&1\\
    \alpha_1&\alpha_2&\dots&\alpha_n\\
    \vdots&\vdots&\vdots&\vdots\\
    \alpha_1^{n-2}&\alpha_2^{n-2}&\dots&\alpha_n^{n-2}\\
    \alpha_1^h&\alpha_2^h&\dots&\alpha_n^h
\end{bmatrix}.$$
The following lemma determines the determinant of the matrix $G_h$.
\begin{lemma}\cite[Proposition 3.1]{Marchi2001}\label{detgh}
With notations as above, for $h\geq n-1$, we have
  $$\det(G_h)=S_{h-n+1}(\alpha_1,\alpha_2,\dots,\alpha_n)\cdot\det(M_{\Lambda}).$$   
\end{lemma}

\begin{lemma}\cite[Lemma 2.6]{Abdukhalikov2025nonrs}\label{paritylemma}
    Let $\Lambda=\{\alpha_1,\alpha_2,\dots,\alpha_n\}\subseteq \mathbb{F}_q$ and $u_i=\prod_{1\leq j\leq n, j\neq i}(\alpha_i-\alpha_j)^{-1}$ for $1\leq i\leq n$. Then 
$$\sum_{i=1}^nu_i\alpha_i^h=\begin{cases}
  0 &\text{ if } 0\leq h\leq n-2\\
  S_{h-n+1}(\alpha_1,\alpha_2,\dots,\alpha_n) &\text{ if }  h\geq n-1
\end{cases}.$$
\end{lemma} 
Throughout the paper, we fix $\Lambda=\{\alpha_1,\alpha_2,\dots,\alpha_n\}\subseteq \mathbb{F}_q$, $S_t=S_t(\alpha_1,\alpha_2,\dots,\alpha_n)$, $\sigma_t=\sigma_t(\alpha_1,\alpha_2,\dots,\alpha_n)$, $u_i$'s as in Eq. \ref{ui}, $D(\alpha_1,\alpha_2,\dots,\alpha_n)=\det(M_{\Lambda})$ unless specified otherwise.

\section{The first family of linear codes}\label{firstfam}
 Let  $n,k$ be positive integers with $3\leq k\leq n-2\leq q-2$, $\Lambda=\{\alpha_1,\alpha_2,\dots,\alpha_n\}\subseteq\mathbb{F}_q$, and $1\leq i<j\leq k$.
 Define the matrix $G_{i,j}(\Lambda)$ by removing the $(i+1)$-th and $(j+1)$-th rows from the generator matrix of $RS(n,k+2,\Lambda)$: 
    $$G_{i,j}(\Lambda)=\begin{bmatrix}
    1 &1& \dots&1\\
    \alpha_1&\alpha_2 & \dots&\alpha_n\\
    \vdots&\vdots&\vdots&\vdots\\
    \alpha_1^{i-1}&\alpha_2^{i-1} & \dots&\alpha_n^{i-1}\\
    \alpha_1^{i+1}&\alpha_2^{i+1} & \dots&\alpha_n^{i+1}\\
    \vdots&\vdots&\vdots&\vdots\\
    \alpha_1^{j-1}&\alpha_2^{j-1} & \dots&\alpha_n^{j-1}\\
    \alpha_1^{j+1}&\alpha_2^{j+1} & \dots&\alpha_n^{j+1}\\
    \vdots&\vdots&\vdots&\vdots\\
    \alpha_1^{k}&\alpha_2^{k} & \dots&\alpha_n^{k}\\
    \alpha_1^{k+1}&\alpha_2^{k+1} & \dots&\alpha_n^{k+1}\\
\end{bmatrix}_{k\times n}.$$
We denote the linear code generated by $G_{i,j}(\Lambda)$ by $C_{i,j}$.
In this section, we give parameters of $C_{i,j}$ and determine a parity check matrix for $C_{i,j}$. We present necessary and sufficient conditions for $C_{i,j}$ to be MDS and investigate the non-GRS property.
\begin{theorem}
    The linear code $C_{i,j}$ has parameters $[n,k,d]$ with $d\in\{n-k-1,n-k,n-k+1\}$.
\end{theorem}
\begin{proof}
   It is easy to see that $C_{i,j}$ is an $[n,k]$ code, since $n>k+1$. Note that $C_{i,j}$ is a subcode of $RS(n,k+2,\Lambda)$. Therefore, $d(C_{i,j})\geq d(RS(n,k+2,\Lambda))=n-(k+2)+1=n-k-1$. Also, by Singleton bound, $d(C_{i,j})\leq n-k+1$. Thus $d(C_{i,j})\in \{n-k-1,n-k,n-k+1\}$.
\end{proof}

\subsection{Parity check matrix for code $C_{i,j}$ }
In this subsection, we determine a parity check matrix for code $C_{i,j}$. For $1\leq m\leq n$ define 
\begin{align*}
    \Gamma_{i,m}=&\sum_{t=0}^{k-i+1}(-1)^{t}\sigma_t \alpha_m^{n-i-t-1}\\
    =& \alpha_m^{n-i-1}-\sigma_1\alpha_m^{n-i-2}+\cdots+ (-1)^{k-i+1}\sigma_{k-i+1}\alpha_m^{n-k-2}.
\end{align*}

\begin{lemma}\label{lemma-ij-parity}
Let $R_i=(u_1\Gamma_{i,1},u_2\Gamma_{i,2},\dots,u_n\Gamma_{i,n})$ with \textcolor{blue}{$1\leq i\leq k$} and  
   $\Lambda^s=(\alpha_1^s,\alpha_2^s,\dots,\alpha_n^s)$. Then for $1\leq s\leq k+1$, $s\neq i$, we have  
   \begin{align*}
      R_i\cdot \Lambda^s= R_i (\Lambda^s)^T=0.
   \end{align*}
\end{lemma}
\begin{proof}
Observe that 
\begin{align*}
      R_i\cdot \Lambda^s=& \sum_{m=1}^n u_m\Gamma_{i,m}\alpha_m^s\\
      =& \sum_{m=1}^n u_m \left (  \sum_{t=0}^{k-i+1}(-1)^{t}\sigma_t \alpha_m^{n-i-t-1}  \right )\alpha_m^s\\
      =& \sum_{m=1}^n \left (  \sum_{t=0}^{k-i+1}(-1)^{t}\sigma_t u_m\alpha_m^{n-i-t-1+s}  \right )\\
      =&  \sum_{t=0}^{k-i+1} (-1)^t \sigma_t \left (  \sum_{m=1}^n u_m\alpha_m^{n-i-t-1+s}  \right )\\
   \text{ (by Lemma \ref{paritylemma})} \hspace{2 cm }   =&   \sum_{t=0}^{k-i+1} (-1)^t \sigma_t S_{(n-i-t-1+s)-(n-1)}\\
   =& \sum_{t=0}^{k-i+1} (-1)^t \sigma_t S_{s-i-t}\\
   =& \sum_{t=0}^{s-i} (-1)^t \sigma_t S_{(s-i)-t}+ \sum_{t=s-i+1}^{k-i+1} (-1)^t \sigma_t S_{(s-i)-t}\\
   =& \sum_{t=0}^{s-i} (-1)^t \sigma_t S_{(s-i)-t}+ 0  \end{align*}
since $(s-i)-t<0\ \ 
\forall t\geq (s-i)+1$.    
Thus, by Eq. (\ref{eqA}), $R_i\cdot \Lambda^s=\sum_{t=0}^{s-i} (-1)^t \sigma_t S_{(s-i)-t}=0$ for $s\neq i$. 
\end{proof}

\begin{theorem}
Let $C_{i,j}$ be the linear code generated by $G_{i,j}(\Lambda)$. Then 
    $$H_{i,j}=\begin{bmatrix}
         u_1&u_2&\dots&u_n\\
u_1\alpha_1&u_2\alpha_2&\dots&u_n\alpha_n\\
u_1\alpha_1^2&u_2\alpha_2^2&\dots&u_n\alpha_n^2\\
\vdots&\vdots&\vdots&\vdots\\
u_1\alpha_1^{n-k-3}&u_2\alpha_2^{n-k-3}&\dots&u_n\alpha_n^{n-k-3}\\
u_1\Gamma_{i,1}&u_2\Gamma_{i,2}&\dots&u_n\Gamma_{i,n}\\
u_1\Gamma_{j,1}&u_2\Gamma_{j,2} &\dots&u_n\Gamma_{j,n}
    \end{bmatrix}_{(n-k)\times n}$$
    is a parity check matrix for $C_{i,j}$.
\end{theorem}
\begin{proof}
Note that the first $(n-k-2)$ rows of $H_{i,j}$ are orthogonal to each row of $G_{i,j}(\Lambda)$ (by Lemma \ref{paritylemma}). Also, the last two rows of $H_{i,j}$ are orthogonal to each row of $G_{i,j}(\Lambda)$ (by Lemma \ref{lemma-ij-parity}). It is sufficient to show that the rows of $H_{i,j}$ are linearly independent.  Denote $R_i=(u_1\Gamma_{i,1},u_2\Gamma_{i,2},\dots,u_n\Gamma_{i,n})$, $R_j=(u_1\Gamma_{j,1},u_2\Gamma_{j,2},\dots,u_n\Gamma_{j,n})$, $h_l=(u_1\alpha_1^l,u_2\alpha_2^l,\dots,u_n\alpha_n^l)$, and  $\Lambda^s=(\alpha_1^s,\alpha_2^s,\dots,\alpha_n^s)$. Then rows of $H_{i,j}$ are $\{h_0,h_1,h_2,\dots,h_{n-k-3},R_i,R_j\}$. 
Note that 
for $0\leq l\leq n-k-3$, 
\begin{align*}
    \Lambda^i \cdot h_l=\sum_{m=1}^nu_m\alpha_m^{l+i}=0
\end{align*}
by Lemma \ref{paritylemma}, since $l+i< n-2$,
\begin{align*}
\Lambda^i\cdot R_i=&\sum_{m=1}^n \alpha_m^i   \left ( u_m \sum_{t=0}^{k-i+1}(-1)^{t}\sigma_t \alpha_m^{n-i-t-1} \right )  \\
=& \sum_{t=0}^{k-i+1}(-1)^{t}\sigma_t   \left (\sum_{m=1}^n u_m  \alpha_m^{n-i-t-1+i} \right ) \\
=& \sum_{t=0}^{k-i+1}(-1)^{t}\sigma_t   \left (\sum_{m=1}^n u_m  \alpha_m^{n-t-1} \right ) \\
=& \sum_{m=1}^n u_m  \alpha_m^{n-1} + \sum_{t=1}^{k-i+1}(-1)^{t}\sigma_t   \left (\sum_{m=1}^n u_m  \alpha_m^{n-t-1} \right ) \\
=& 1+0=1. \\
\end{align*}
Similarly, $\Lambda^j\cdot R_j =1$. 
Furthermore, 
\begin{align*}
    \Lambda^j\cdot R_i=&\sum_{m=1}^n\alpha_m^j  \left ( u_m \sum_{t=0}^{k-i+1}(-1)^{t}\sigma_t \alpha_m^{n-i-t-1} \right )\\
    =&\sum_{t=0}^{k-i+1}(-1)^{t}\sigma_t \left (\sum_{m=1}^n u_m \alpha_m^{n-i-t-1+j} \right )\\  =&\sum_{t=0}^{k-i+1}(-1)^{t}\sigma_t S_{j-i-t}\\ 
       =&\sum_{t=0}^{j-i}(-1)^{t}\sigma_t S_{j-i-t}+\sum_{t=j-i+1}^{k-i+1}(-1)^{t}\sigma_t S_{j-i-t}=0\\ 
\end{align*}
by Eq. (\ref{eqA}), since $j-i\geq 1$.
Now, we show that the rows of $H_{i,j}$ are linearly independent.
Suppose that there exist scalars $a_0,a_1,\dots,a_{n-k-3},$ $a_{n-k-2},a_{n-k-1}$ such that 
\begin{align*}
    a_0h_0+a_1h_1+\cdots+a_{n-k-3}h_{n-k-3}+a_{n-k-2}R_i+a_{n-k-1}R_j=0.
\end{align*}
Then, \begin{align*}
    &\Lambda^j\cdot (a_0h_0+a_1h_1+\cdots+a_{n-k-3}h_{n-k-3}+a_{n-k-2}R_i+a_{n-k-1}R_j)=0\\
    & \implies  a_{n-k-1}=0,
\end{align*}
and \begin{align*}
    &\Lambda^i\cdot (a_0h_0+a_1h_1+\dots+a_{n-k-3}h_{n-k-3}+a_{n-k-2}R_i+a_{n-k-1}R_j)=0\\
    & \implies  a_{n-k-2}=0.
\end{align*}
Therefore,  $a_0h_0+a_1h_1+\cdots+a_{n-k-3}h_{n-k-3}=0$. It is easy to see that rows $h_0,h_1,\dots,h_{n-k-3}$ are linearly independent. Hence $a_0=a_1=\cdots=a_{n-k-3}=0$. Thus, rows of $H_{i,j}$ are linearly independent.

% For $0\leq l\leq n-k-3$, 
% \begin{align*}
%     \Lambda^i h_l^T=\sum_{m=1}^nu_m\alpha_m^{l+i}=0
% \end{align*}
% by Lemma \ref{paritylemma}, since $l+i< n-2$. But 
% \begin{align*}
% \Lambda^iR_i^T=&\sum_{m=1}^n \alpha_m^i   \left ( u_m \sum_{t=0}^{k-i+1}(-1)^{t}\sigma_t \alpha_m^{n-i-t-1} \right )  \\
% =& \sum_{t=0}^{k-i+1}(-1)^{t}\sigma_t   \left (\sum_{m=1}^n u_m  \alpha_m^{n-i-t-1+i} \right ) \\
% =& \sum_{t=0}^{k-i+1}(-1)^{t}\sigma_t   \left (\sum_{m=1}^n u_m  \alpha_m^{n-t-1} \right ) \\
% =& \sum_{m=1}^n u_m  \alpha_m^{n-1} + \sum_{t=1}^{k-i+1}(-1)^{t}\sigma_t   \left (\sum_{m=1}^n u_m  \alpha_m^{n-t-1} \right ) \\
% =& 1+0=1. \\
% \end{align*}
% Therefore, $R_i\notin \text{span}\{h_0,h_1,h_2,\dots,h_{n-k-3}\}$. Using similar arguments, we get 
% \begin{align*}
%   \Lambda^jh_l^T=0 \ \ \forall\  0\leq l\leq n-k-3 \text{ and } \Lambda^jR_j^T=1.\\ 
% \end{align*}
% Also, 
% \begin{align*}
%     \Lambda^jR_i^T=&\sum_{m=1}^n\alpha_m^j  \left ( u_m \sum_{t=0}^{k-i+1}(-1)^{t}\sigma_t \alpha_m^{n-i-t-1} \right )\\
%     =&\sum_{t=0}^{k-i+1}(-1)^{t}\sigma_t \left (\sum_{m=1}^n u_m \alpha_m^{n-i-t-1+j} \right )\\   
%      =&\sum_{t=0}^{k-i+1}(-1)^{t}\sigma_t S_{j-i-t}\\ 
%        =&\sum_{t=0}^{j-i}(-1)^{t}\sigma_t S_{j-i-t}+\sum_{t=j-i+1}^{k-i+1}(-1)^{t}\sigma_t S_{j-i-t}=0\\ 
% \end{align*}
% by Lemma \ref{paritylemma}, since $j-i\geq 1$. Therefore,  $R_j\notin \text{span}\{h_0,h_1,h_2,\dots,h_{n-k-3}, R_i\}$. Hence, the rows of $H_{i,j}$ are linearly independent.
\end{proof}

\subsection{MDS and non-GRS properties of code $C_{i,j}$ }
In this subsection, we establish the necessary and sufficient conditions under which the code $C_{i,j}$ is MDS and analyze the non-GRS property of this family. We also provide explicit constructions and examples of non-GRS MDS codes.
To derive the necessary and sufficient condition for MDSness, we begin with the following proposition, which serves as a key tool in the proof.

\begin{proposition}\label{det-ij}
Let $\Lambda=\{\alpha_1,\alpha_2,\dots,\alpha_n\}\subseteq\mathbb{F}_q$ and for $1\leq i<j\leq n-1$, $$V_{i,j}(\Lambda)=\begin{bmatrix}
    1 &1& \dots&1\\
    \alpha_1&\alpha_2 & \dots&\alpha_n\\
    \vdots&\vdots&\vdots&\vdots\\
    \alpha_1^{i-1}&\alpha_2^{i-1} & \dots&\alpha_n^{i-1}\\
    \alpha_1^{i+1}&\alpha_2^{i+1} & \dots&\alpha_n^{i+1}\\
    \vdots&\vdots&\vdots&\vdots\\
    \alpha_1^{j-1}&\alpha_2^{j-1} & \dots&\alpha_n^{j-1}\\
    \alpha_1^{j+1}&\alpha_2^{j+1} & \dots&\alpha_n^{j+1}\\
    \vdots&\vdots&\vdots&\vdots\\
    \alpha_1^{n}&\alpha_2^{n} & \dots&\alpha_n^{n}\\
    \alpha_1^{n+1}&\alpha_2^{n+1} & \dots&\alpha_n^{n+1}\\
\end{bmatrix}_{n\times n}.$$
Then
\begin{align*}
    &\det (V_{i,j}(\Lambda))=(-1)^{-(i+j-1)}D\cdot [\sigma_{n-i} \sigma_{n-j+1}-\sigma_{n-i+1}\sigma_{n-j}]
\end{align*}
where $D=D(\alpha_1,\alpha_2,\dots,\alpha_n)$ is the determinant of Vandermonde matrix $M_{\Lambda}$ and $\sigma_t=\sigma_t(\alpha_1,\alpha_2,\dots,\alpha_n)$.
\end{proposition}

\begin{proof}
  Let $$M_{\Lambda}(x,y)=\begin{bmatrix}
    1 &1& \dots&1 &1 &1\\
    \alpha_1&\alpha_2 & \dots&\alpha_n& x &y\\
    \vdots&\vdots&\vdots&\vdots &\vdots &\vdots\\
    \alpha_1^{i-1}&\alpha_2^{i-1} & \dots&\alpha_n^{i-1} & x^{i-1} & y^{i-1}\\
    \alpha_1^{i}&\alpha_2^{i} & \dots&\alpha_n^{i} &x^i &y^i\\
    \alpha_1^{i+1}&\alpha_2^{i+1} & \dots&\alpha_n^{i+1}& x^{i+1} & y^{i+1}\\
    \vdots&\vdots&\vdots&\vdots &\vdots &\vdots\\
    \alpha_1^{j-1}&\alpha_2^{j-1} & \dots&\alpha_n^{j-1}& x^{j-1} & y^{j-1}\\
    \alpha_1^{j}&\alpha_2^{j} & \dots&\alpha_n^{j}& x^{j} & y^{j}\\
    \alpha_1^{j+1}&\alpha_2^{j+1} & \dots&\alpha_n^{j+1}& x^{j+1} & y^{j+1}\\
    \vdots&\vdots&\vdots&\vdots&\vdots &\vdots\\
    \alpha_1^{n}&\alpha_2^{n} & \dots&\alpha_n^{n}& x^{n} & y^{n}\\
    \alpha_1^{n+1}&\alpha_2^{n+1} & \dots&\alpha_n^{n+1}& x^{n+1} & y^{n+1}\\
\end{bmatrix}_{(n+2)\times (n+2)}.$$  
Then by denoting $\alpha_{n+1}=x$, $\alpha_{n+2}=y$, we have
\begin{align*}
\det(M_{\Lambda}(x,y))=&\prod_{1\leq l<m\leq n+2}(\alpha_m-\alpha_l)\\ 
=&\prod_{1\leq l<m\leq n}(\alpha_m-\alpha_l)\cdot \prod_{1\leq l\leq n}(x-\alpha_l)\cdot\prod_{1\leq l\leq n}(y-\alpha_l)\cdot(y-x)\\
=& D \cdot (y-x) \cdot \prod_{1\leq l\leq n}(x-\alpha_l)\cdot\prod_{1\leq l\leq n}(y-\alpha_l)\\
=& D \cdot (y-x) \cdot \left (\sum_{l=0}^n(-1)^l\sigma_l x^{n-l}\right)\cdot \left (\sum_{l=0}^n(-1)^l\sigma_l y^{n-l}\right) \\
=& D \cdot (y-x) \cdot A(x) \cdot A(y) \ \ \ (say) \\
=& D \cdot (y-x) \cdot B(x,y), \\
\end{align*}
where $B(x,y)=A(x) \cdot A(y)$.
 Note that the coefficient of $x^ly^m$ in $B(x,y)$ is $(-1)^{2n-(l+m)}\sigma_{n-l}\sigma_{n-m}$. Therefore, the coefficient of $x^iy^j$ in $\det(M_{\Lambda}(x,y))$ is given by
 \begin{align*}
   &\text{coefficient of } x^iy^{j-1} \text{ in } B(x,y) - \text{coefficient of } x^{i-1}y^{j} \text{ in } B(x,y)\\&= D[ (-1)^{2n-(i+j-1)} \sigma _{n-i}\sigma_{n-j+1}-(-1)^{2n-(i-1+j)} \sigma _{n-i+1}\sigma_{n-j}]\\
   &=D(-1)^{2n-(i+j-1)} [\sigma _{n-i}\sigma_{n-j+1}- \sigma _{n-i+1}\sigma_{n-j}]\\
    &=D(-1)^{-(i+j-1)} [\sigma _{n-i}\sigma_{n-j+1}- \sigma _{n-i+1}\sigma_{n-j}].
 \end{align*}
Also, observe that $\det(V_{i,j}(\Lambda))$ is coefficient of $x^iy^j$ in $\det(M_{\Lambda}(x,y))$. This completes the proof.
\end{proof}

 \begin{theorem} \label{MDS-ij}
    The linear code $C_{i,j}$ generated by $G_{i,j}(\Lambda)$ is MDS if and only if
$$\sigma_{k-i}(\beta_1,\beta_2,\dots,\beta_k)\sigma_{k-j+1}(\beta_1,\beta_2,\dots,\beta_k)-\sigma_{k-i+1}(\beta_1,\beta_2,\dots,\beta_k)\sigma_{k-j}(\beta_1,\beta_2,\dots,\beta_k)\neq 0$$
 for all $\{\beta_1,\beta_2,\dots,\beta_k\}\subseteq \Lambda$ with distinct $\beta_i$.
\end{theorem}
\begin{proof}
    The proof follows from Lemma \ref{mdsnonsingular} and Proposition \ref{det-ij}.
\end{proof}
    
\begin{proposition}\label{k=34nongrs}
(a). Let $k=3$ , $k<(n-2)/2$, $q>9$ and code $C_{1,2}$ (resp.  $C_{2,3}$) be an MDS code. Then $C_{1,2}$ (resp.  $C_{2,3}$)  is a non-GRS MDS code.\\
(b). Let $k=4$, $k<(n-2)/2$, $q>11$ and code $C_{i,j}$ be an MDS code for $1\leq i<j\leq k$. Then code $C_{i,j}$ is a non-GRS MDS code.
\end{proposition}
\begin{proof}
(a). Let $k=3$. The Schur square of $C_{1,2}$ is generated by \begin{align*}
    T=\{(\alpha_1^l,\alpha_2^l,\dots,\alpha_n^l): l=0,3,4,6,7,8\}.
\end{align*}
The Schur square of $C_{2,3}$ is generated by \begin{align*}
    T=\{(\alpha_1^l,\alpha_2^l,\dots,\alpha_n^l): l=0,1,2,4,5,8\}.
\end{align*}
Also, $8< q-2$, thus $T$ is linearly independent. Hence, $\dim(C_{1,2}\star C_{1,2})=\dim(C_{2,3}\star C_{2,3})=2k$.\\
(b). Let $k=4$. 
The Schur square of $C_{1,2}$ is generated by \begin{align*}
    T=\{(\alpha_1^l,\alpha_2^l,\dots,\alpha_n^l): l=0,3,4,5,6,7,8,9,10\}.
\end{align*}

The Schur square of $C_{1,3}$ is generated by \begin{align*}
    T=\{(\alpha_1^l,\alpha_2^l,\dots,\alpha_n^l): i=0,2,4,5,6,7,8,9,10\}.
\end{align*}

 The Schur square of $C_{1,4}$ is generated by \begin{align*}
    T=\{(\alpha_1^l,\alpha_2^l,\dots,\alpha_n^l): l=0,2,3,4,5,6,7,8,10\}.
\end{align*}

The Schur square of $C_{2,3}$ is generated by \begin{align*}
    T=\{(\alpha_1^l,\alpha_2^l,\dots,\alpha_n^l): l=0,1,2,4,5,6,8,9,10\}.
\end{align*}

 The Schur square of $C_{2,4}$ is generated by \begin{align*}
    T=\{(\alpha_1^l,\alpha_2^l,\dots,\alpha_n^l): l=0,1,2,3,4,5,6,8,10\}.
\end{align*}

 The Schur square of $C_{3,4}$ is generated by \begin{align*}
    T=\{(\alpha_1^l,\alpha_2^l,\dots,\alpha_n^l): l=0,1,2,3,4,5,6,7,10\}.
\end{align*}
As $9\leq  q-2$, therefore, $\dim(C_{i,j}\star C_{i,j}) \ge 2k$ for $1\leq i<j\leq k$.\\
Hence, by Lemma \ref{lemmanongrs},  the codes are non-GRS.
\end{proof}

\begin{theorem}\label{thm-nongrs-ij}
Let $5\leq k\leq \frac{n-4}{2}$ and $C_{i,j}$ be an MDS code generated by $G_{i,j}$. Then $C_{i,j}$ is a non-GRS MDS code. 
\end{theorem}
\begin{proof}
We will show that 
   \begin{align*}
       \dim (C_{i,j}\star C_{i,j})=\begin{cases}
           2k+1 & \text{ if } i=1,j=2,3,k\\
           2k+2 & \text{ if } i=1, 4\leq j\leq k-1\\
           2k+2 & \text{ if } i=2,j=3,k\\
           2k+3 & \text{ if } 2\leq i\leq k-3, 4\leq j\leq k-1\\
           2k+2 & \text{ if } i=k-2,j=k-1\\
        2k+1 & \text{ if } i=k-2, k-1, j=k.\\
       \end{cases}
   \end{align*}
Let $i=1$ and $j=2,3$. Then 
\begin{align*}
    C_{i,j}\star C_{i,j}=\text{span}\{(\alpha_1^l,\alpha_2^l,\dots,\alpha_n^l)|\ 0\leq l\leq 2k+2, l\neq 1,j\}.
\end{align*}
Let $i=1$ and $j=k$. Then 
\begin{align*}
    C_{i,j}\star C_{i,j}=\text{span}\{(\alpha_1^l,\alpha_2^l,\dots,\alpha_n^l)|\ 0\leq l\leq 2k+2, l\neq 1,2k+1\}.
\end{align*}

%Let $i=1$ and $j=3$. Then 
%\begin{align*}
%    C_{i,j}\star C_{i,j}=\text{span}\{(\alpha_1^l,\alpha_2^l,\dots,\alpha_n^l)|\ 0\leq l\leq 2k+2, l\neq 1,3\}.
%\end{align*}
Let $i=1$ and $4\leq j \leq k-1$. Then 
\begin{align*}
    C_{i,j}\star C_{i,j}=\text{span}\{(\alpha_1^l,\alpha_2^l,\dots,\alpha_n^l)|\ 0\leq l\leq 2k+2, l\neq 1 \}.
\end{align*}
Let $i=2$ and $j=3$. Then 
\begin{align*}
    C_{i,j}\star C_{i,j}=\text{span}\{(\alpha_1^l,\alpha_2^l,\dots,\alpha_n^l)|\ 0\leq l\leq 2k+2, l\neq 3\}.
\end{align*}
Let $i=2$ and $j=k$. Then 
\begin{align*}
    C_{i,j}\star C_{i,j}=\text{span}\{(\alpha_1^l,\alpha_2^l,\dots,\alpha_n^l)|\ 0\leq l\leq 2k+2, l\neq 2k+1\}.
\end{align*}
Let $2\leq i\leq k-3$ and $4\leq j\leq k-1$. Then 
\begin{align*}
    C_{i,j}\star C_{i,j}=\text{span}\{(\alpha_1^l,\alpha_2^l,\dots,\alpha_n^l)|\ 0\leq l\leq 2k+2\}.
\end{align*}
Let $i=k-2$ and $j= k-1$. Then 
\begin{align*}
    C_{i,j}\star C_{i,j}=\text{span}\{(\alpha_1^l,\alpha_2^l,\dots,\alpha_n^l)|\ 0\leq l\leq 2k+2, l\neq 2k-1\}.
\end{align*}
Let $i=k-2$ and $j= k$. Then 
\begin{align*}
    C_{i,j}\star C_{i,j}=\text{span}\{(\alpha_1^l,\alpha_2^l,\dots,\alpha_n^l)|\ 0\leq l\leq 2k+2, l\neq 2k-1, \ l\neq 2k+1\}.
\end{align*}
Let $i=k-1$ and $j= k$. Then 
\begin{align*}
    C_{i,j}\star C_{i,j}=\text{span}\{(\alpha_1^l,\alpha_2^l,\dots,\alpha_n^l)|\ 0\leq l\leq 2k+2, l\neq 2k, \ l\neq 2k+1\}.
\end{align*}
Thus, by Lemma \ref{lemmanongrs}, $C_{i,j}$ is a non-GRS MDS code.
\end{proof}

Using techniques similar to those in  \cite{Jin2024}, we now present an explicit construction of non-GRS MDS codes,  
demonstrating that such codes can be constructed with lengths up to $\left (\frac{q}{p}\right )^{1/4}$.
\begin{theorem}\label{consk1k2} Let $5\leq k\leq \frac{n-2}{2}$, $q=p^m$, $\mathbb{F}_q=\mathbb{F}_p[\gamma]$ and $t=\lfloor \frac{m-1}{4}\rfloor$. Then for $n\leq p^{t}$, if $p\nmid \frac{1}{12}k^2(k^2-1)$, then code $C_{k-2,k-1}$ is a non-GRS MDS code over $\mathbb{F}_q$.
\end{theorem}
\begin{proof}
   Let $V_t=\{g(x)\in \mathbb{F}_p[x]| g(x) \text{ is monic and } \deg g(x)= t\}$. Then the cardinality of $V_t$ is $p^{t}$. We enumerate the elements of $V_t$ as $g_1(x),g_2(x),\dots,g_{p^{t}}(x)$. Observe that $g_r(\gamma)\neq g_s(\gamma)$ for $r\neq s$. Take $\Lambda=\{\alpha_r=g_r(\gamma) :  1\leq r\leq n\}$. Then for any $ \{r_1,r_2,\dots,r_{k}\}\subset \{1,2,\dots,n\}$, let  $\bm{\beta}_k=(\alpha_{r_1},\alpha_{r_2},\dots,\alpha_{r_k})$,
\begin{align*}
    \begin{split}
         &\sigma_2^2(\bm{\beta_k})-\sigma_1(\bm{\beta_k})\sigma_3(\bm{\beta_k})=\left(\sum_{1\leq s_1<s_2\leq k}\alpha_{r_{s_1}}\alpha_{r_{s_2}}\right)^2-\left(\sum_{1\leq s_1\leq k}\alpha_{r_{s_1}}\right)\left(\sum_{1\leq s_1<s_2<s_3\leq k}\alpha_{r_{s_1}}\alpha_{s_{s_2}}\alpha_{r_{s_3}}\right)\\
         &=\left(\sum_{1\leq s_1<s_2\leq k}g_{r_{s_1}}(\gamma)g_{r_{s_2}}(\gamma)\right)^2-\left(\sum_{1\leq s_1\leq k}g_{r_{s_1}}(\gamma)\right)\left(\sum_{1\leq s_1<s_2<s_3\leq k}g_{r_{s_1}}(\gamma)g_{r_{s_2}}(\gamma)g_{r_{s_3}}(\gamma)\right)\\
         &\neq 0,
         \end{split}
\end{align*}
as  the polynomial  \begin{align*}
  \left(\sum_{1\leq s_1<s_2\leq k}g_{r_{s_1}}(x)g_{r_{s_2}}(x)\right)^2-\left(\sum_{1\leq s_1\leq k}g_{r_{s_1}}(x)\right)\left(\sum_{1\leq s_1<s_2<s_3\leq k}g_{r_{s_1}}(x)g_{r_{s_2}}(x)g_{r_{s_3}}(x)\right)
     \end{align*} 
 has degree $4t<m$ with leading coefficient $\left (\binom{k}{2}^2-k\binom{k}{3}\right)=\frac{1}{12}k^2(k^2-1)$.
 Hence, by Theorem \ref{MDS-ij}, Proposition \ref{k=34nongrs}, and Theorem \ref{thm-nongrs-ij}, the code $C_{k-2,k-1}$ generated by $G_{k-2,k-1}$ is an $[n,k]$ non-GRS MDS code. 
\end{proof}

\begin{remark}
Theorem \ref{consk1k2} is true for $k=3,4$ in accordance with Proposition \ref{k=34nongrs}.
\end{remark}

The next example shows that the matrix $G_{i,j}$ can be utilized to construct non-GRS MDS codes of length $(q-1)/2$.
\begin{example}
    Let $q=17, n=9, k=3$ and $\{\alpha_1,\alpha_2,\dots,\alpha_9\}=\{0,2, 3, 4, 5, 7, 9,10,16\}\subset \mathbb{F}_q$. Then the linear code $C_{1,2}$ generated by
    $$G_{1,2}=\begin{bmatrix}
        1&1&1&1&1&1&1&1&1\\
        0&2^3&3^3&4^3&5^3&7^3&9^3&10^3&16^3\\
         0&2^4&3^4&4^4&5^4&7^4&9^4&10^4&16^4\\
    \end{bmatrix}$$
    is a non-GRS MDS code. 
\end{example}

\subsection{Self-orthogonality and Self-duality  of codes $C_{i,j}$}
In this subsection, we investigate the self-orthogonality and self-duality of the code $C_{i,j}$. The following lemma is an easy observation.

\begin{lemma}\label{selflemma}
    Let $C$ be an $[n,k]$ linear code with a generator matrix $G$. Then $C$ is self-orthogonal if and only if $GG^T$ is a $k\times k$ zero matrix. Moreover, if $n=2k$, then a self-orthogonal code is a self-dual code.
\end{lemma}

Let $\bm{v}=(v_1,v_2,\dots,v_n)\in (\mathbb{F}_q^*)^{n}$. Consider a matrix $$G_{i,j,\bm{v}}=G_{i,j}\cdot \text{diag}[v_1,v_2,\dots,v_n].$$
We denote the linear code generated by $G_{i,j,\bm{v}}$ by $C_{i,j,\bm{v}}$. Then the code $C_{i,j,\bm{v}}$ and $C_{i,j}$ are equivalent. Moreover, $C_{i,j,\bm{v}}$ is MDS if and only if $C_{i,j}$ is MDS.

% Let $q$ be odd and let $\chi$ be the quadratic character of $\mathbb{F}_q^*$ (for details, see \cite{lidl1997}). Then
% \begin{equation}
%     \chi(x)=\begin{cases}
%         1 & \text{ if } x \text{ is a square in } \mathbb{F}_q^*\\
%         -1 & \text{ otherwise }
%     \end{cases}.
% \end{equation}
% The lemma is well known. We give a proof for the completeness.
% \begin{lemma}
% Let $(u_1,u_2,\dots,u_n)\in \mathbb{F}_q^*$. If $\chi(u_i)$ is same for all $1\leq i\leq n$, then there exists $(v_1,v_2,\dots,v_n)\in \mathbb{F}_q^*$ such that $v_i^2=\lambda u_i$ for some $\lambda\in \mathbb{F}_q^*$.    
% \end{lemma}
% \begin{proof}
    
% \end{proof}

\begin{theorem}\label{thm-so-ij}
Let $ 5\leq k\leq \frac{n-4}{2}$, $2\leq i\leq k-3$, $4\leq j\leq k-1$, ${\bm{v}}=(v_1,v_2,\dots,v_n)\in (\mathbb{F}_q^*)^{n} $, and $u_i$, $1\leq i\leq n$ as in Eq. \ref{ui}. Then the code $C_{i,j,\bm{v}}$ is self-orthogonal if and only if there exists a polynomial $f(x)=f_0+f_1x+\cdots+f_{n-2k-2}x^{n-2k-4}\in \mathbb{F}_q[x]$ such that $v_m^2=u_mf(\alpha_m)$, $1\leq m\leq n$.
\end{theorem}
\begin{proof}
   By Lemma \ref{selflemma}, $C_{i,j,\bm{v}}$ is self-orthogonal if and only if $G_{i,j,\bm{v}}G_{i,j,\bm{v}}^T=0$, that is,
    \begin{align}
  \label{so1-ij}   &\sum_{m=1}^nv_m^2\alpha_m^l=0  \text{ for } 0\leq l\leq 2k+2. 
   \end{align}
Let $C_{i,j,\bm{v}}$ be self-orthogonal, that is, Eq. (\ref{so1-ij}) holds true. Then, Eq. \ref{so1-ij} implies that  $(v_1^2,v_2^2,\dots,v_n^2)$ is a nonzero solution of the following system of linear equations

 \begin{align} \label{so5-ij}
           &\sum_{m=1}^n\alpha_m^l x_m=0  \text{ for all}\ l \text{ with}\ 0\leq l\leq 2k+2. 
        \end{align}
The rank of the coefficient matrix of Eq. \ref{so5-ij} is $2k+3$. Therefore, the dimension of the solution space of Eq. \ref{so5-ij} is $n-2k-3$. Also, by Lemma \ref{paritylemma},
\begin{align*}
\begin{split}
    (u_1,u_2,\dots,u_n), (u_1\alpha_1,u_2\alpha_2,\dots,u_n\alpha_n),\dots,(u_1\alpha_1^{n-2k-4},u_2\alpha_2^{n-2k-4},\dots,u_n\alpha_n^{n-2k-4})
    \end{split}
    \end{align*}
   are $n-2k-3$ linearly independent solutions of Eq. \ref{so5-ij}. This implies that there exist scalars $f_0,f_1,\dots,f_{n-2k-4}$ such that 
   \begin{align*}
   \begin{split}
       v_m^2=\sum_{t=0}^{n-2k-4} u_mf_t\alpha_m^t.
        \end{split}
   \end{align*}
Denote $f(x)=\sum_{t=0}^{n-2k-4}f_tx^t$. Then  $v_m^2= u_mf(\alpha_m)$ for $1\leq m\leq n$.

Conversely, let there exists a polynomial $f(x)=\sum_{t=0}^{n-2k-4}f_tx^t$ such that $v_m^2=u_mf(\alpha_m)$, $1\leq m\leq n$.  It is easy to see that $(v_1^2,v_2^2,\dots,v_n^2)$ will satisfy Eq. \ref{so1-ij}. Thus, $C_{i,j,\bm{v}}$ is self-orthogonal. 
\end{proof}

\begin{theorem}\label{thmsok1k2}
Let $5\leq k\leq  \frac{n}{2}$, $i=k-2$, $j=k-1$, ${\bm{v}}=(v_1,v_2,\dots,v_n)\in (\mathbb{F}_q^*)^{n} $, and $u_m$, $1\leq m\leq n$ as in Eq. \ref{ui}. Then the code $C_{i,j,\bm{v}}$ is self-orthogonal if and only if there exists a polynomial $f(x)=f_0+f_1x+\cdots+f_{n-2k}x^{n-2k}\in \mathbb{F}_q[x]$ such that the following hold.
 \begin{enumerate}
     \item $v_m^2=u_mf(\alpha_m)$, $1\leq m\leq n$.\label{1}
     \item $f_{n-2k-1}+f_{n-2k}S_1=0$.\label{2}
     \item $f_{n-2k-2}+f_{n-2k-1}S_1+f_{n-2k}S_2=0$.\label{3}
     \item $f_{n-2k-3}+f_{n-2k-2}S_1+f_{n-2k-1}S_2+f_{n-2k}S_3=0$.\label{4}
 \end{enumerate}
  Where $f_t=0$ for $t<0$.
\end{theorem}
\begin{proof}
   By Lemma \ref{selflemma}, $C_{i,j,\bm{v}}$ is self-orthogonal if and only if $G_{i,j,\bm{v}}G_{i,j\bm{v}}^T=0$, that is,
    \begin{align}
        \label{so1}    &\sum_{m=1}^nv_m^2\alpha_m^l=0  \text{ for } 0\leq l\leq 2k-2 \\
          \label{so2}  &\sum_{m=1}^nv_m^2\alpha_m^{2k}=0&  \\
          \label{so3}  &\sum_{m=1}^nv_m^2\alpha_m^{2k+1}=0&\\
          \label{so4}  &\sum_{m=1}^nv_m^2\alpha_m^{2k+2}=0.&
   \end{align}
Let $C_{i,j,\bm{v}}$ be self-orthogonal, that is, Eqs. \ref{so1}-\ref{so4} hold. Then, Eq. \ref{so1} implies that  $(v_1^2,v_2^2,\dots,v_n^2)$ is a nonzero solution of the following system of linear equations

 \begin{align} \label{so5}
           &\sum_{m=1}^n\alpha_m^l x_m=0  \text{ for all}\ l \text{ with}\ 0\leq l\leq 2k-2. 
        \end{align}
The rank of the coefficient matrix of Eq. \ref{so5} is $2k-1$. Therefore, the dimension of the solution space of Eq. \ref{so5} is $n-2k+1$. Also, by Lemma \ref{paritylemma},
\begin{align*}
\begin{split}
    (u_1,u_2,\dots,u_n), (u_1\alpha_1,u_2\alpha_2,\dots,u_n\alpha_n),\dots,(u_1\alpha_1^{n-2k},u_2\alpha_2^{n-2k},\dots,u_n\alpha_n^{n-2k})
    \end{split}
    \end{align*}
   are $n-2k+1$ linearly independent solutions of Eq. \ref{so5}. This implies that there exist scalars $f_0,f_1,\dots,f_{n-2k}$ such that 
   \begin{align*}
   \begin{split}
     % &(v_1^2,v_2^2,\dots,v_n^2)=\\ 
     % &f_0(u_1,u_2,\dots,u_n)+ f_1(u_1\alpha_1,u_2\alpha_2,\dots,u_n\alpha_n)+\dots+f_{n-2k}(u_1\alpha_1^{n-2k},u_2\alpha_2^{n-2k},\dots,u_n\alpha_n^{n-2k})\\
       v_m^2=\sum_{t=0}^{n-2k} u_mf_t\alpha_m^t.
        \end{split}
   \end{align*}
Denote $f(x)=\sum_{t=0}^{n-2k}f_tx^t$. Then  $v_m^2= u_mf(\alpha_m)$ for $1\leq m\leq n$. Furthermore, $v_m^2$, $1\leq m\leq n$, also satisfy Eqs. \ref{so2}, \ref{so3}, \ref{so4}. Thus, by Lemma \ref{paritylemma}, we have \ref{2}, \ref{3} and \ref{4}.

Conversely, let there exists a polynomial $f(x)=\sum_{t=0}^{n-2k}f_tx^t$ satisfying \ref{1}-\ref{4}. It is easy to see that $(v_1^2,v_2^2,\dots,v_n^2)$ will satisfy Eqs. \ref{so1}-\ref{so4}. Thus, $C_{i,j,\bm{v}}$ is self-orthogonal. 
\end{proof}
In the following example, we construct a self-orthogonal code using Theorem \ref{thmsok1k2}.

\begin{example}
 Let $n=11$, $k=5$, $q=2^5$ and $\mathbb{F}_q=\mathbb{F}_2[w]$ with $w^5+w^2+1=0$. Let $\Lambda=\{w,w^2,w^3, w^4, w^5, w^6, w^{10}, w^{13}, w^{17}, w^{21}, w^{26}\}\subseteq\mathbb{F}_q$. Then 
 \begin{align*}
     u_1=w^4, u_2=w^{29}, u_3=w, u_4=w^{14}, u_5=w^{18},u_6=w^{25},\\
      u_7=w^{11}, u_8={w}, u_9=w^{10}, u_{10}=w, u_{11}=w^6.
 \end{align*}
Let $f(x)=x$. Then $u_mf(\alpha_m)=u_m\alpha_m$, $1\leq m\leq n$, are squares. As per Theorem \ref{thmsok1k2}, take $v_m$'s such that $v_m^2=u_mf(\alpha_m)$, $1\leq m\leq n$, that are, \begin{align*}
    v_1=w^{18}, v_2=1, v_3=w^2,v_4=w^9,v_5=w^{27},v_6=1,\\
    v_7=w^{26},v_8=w^7,v_9=w^{29},v_{10}=w^{11},v_{11}=w^{16}.
\end{align*} 
It is easy to see that $\alpha_m$'s satisfy the conditions in Theorem \ref{thmsok1k2} for $f(x)=x$. Therefore,  $C_{k-2,k-1,\bm{v}}$ is a self-orthogonal code over $\mathbb{F}_q$ with parameters $[11,5,6]$ (AMDS) and a generator matrix 
\begin{align*}
    G_{k-2,k-1,\bm{v}}&=G_{k-2,k-1}\cdot \text{diag}(\bm{v})\\
    &=\begin{bmatrix}
        w^{18}&    1&  w^2&  w^9 &w^{27} &   1 &w^{26}&  w^7 &w^{29} &w^{11}& w^{16}\\
w^{19}&  w^2&  w^5& w^{13} &   w & w^6 & w^5& w^{20} &w^{15}&    w& w^{11}\\
w^{20}&  w^4&  w^8& w^{17}&  w^6& w^{12}& w^{15}&  w^2&    w &w^{22} & w^6\\
w^{23}& w^{10} &w^{17}& w^{29}& w^{21}& w^{30}& w^{14}& w^{10}& w^{21}& w^{23}& w^{22}\\
w^{24}& w^{12}& w^{20} & w^2& w^{26} & w^5& w^{24} &w^{23} & w^7& w^{13} &w^{17}
    \end{bmatrix}.
\end{align*}
\end{example}

Using Theorem \ref{thm-so-ij} and Theorem \ref{thmsok1k2}, we have the following characterization for self-duality.

\begin{corollary}\label{cor-so-ij}
Let $ 5\leq k\leq \frac{n-4}{2}$, $2\leq i\leq k-3$, $4\leq j\leq k-1$. Then the code $C_{i,j,\bm{v}}$ is never self-dual  for any $\bm{v}\in (\mathbb{F}_q^*)^n$.
\end{corollary}
\begin{proof}
For $n=2k$, we have $n-2k-4<0$. Thus, by Theorem \ref{thm-so-ij}, $C_{i,j,\bm{v}}$ is not self-orthogonal for any $\bm{v}\in (\mathbb{F}_q^*)^n$.    
\end{proof}

\begin{corollary}\label{thmsdk1k2}
Let $k\geq 5$, $n=2k$, $\{\alpha_1,\alpha_2,\dots,\alpha_n\}\subseteq\mathbb{F}_q$, and $u_m$, $1\leq m\leq n$ as in Eq. \ref{ui}. Then $C_{k-2,k-1,\bm{v}}$ is self-dual for some ${\bm{v}}=(v_1,v_2,\dots,v_n)\in(\mathbb{F}_q^*)^n$ if and only if  the following hold:
\begin{enumerate}
    \item $v_m^2=\lambda u_m$, $1\leq m\leq n$, for some $\lambda\in \mathbb{F}_q^*$.
    \item $S_1=S_2=S_3=0$.
\end{enumerate}
\end{corollary}
\begin{proof}
By putting $n=2k$ in Theorem \ref{thmsok1k2}, we get the desired result.    
\end{proof}

\begin{example}
 Let $n=10$, $k=5$, $q=2^4$ and $\mathbb{F}_q=\mathbb{F}_2[w]$ with $w^4+w+1=0$. For $\Lambda=\{w, w^2, w^4, w^5, w^7, w^8, w^{10}, w^{11}, w^{13}, w^{14}\}\subseteq\mathbb{F}_q$, we have $S_1=S_2=S_3=0$ and 
 \begin{align*}
     u_1=w^{11}, u_2=w^7,u_3=w^{14},u_4=w^{10},u_5=w^2,u_6=w^{13},u_7=w^5,u_8=w,u_9=w^8,u_{10}=w^4.
 \end{align*}
Take \begin{align*}
    v_1=w^{13},  v_2=w^{11},v_3=w^{7},v_4=w^{5},v_5=w,v_6=w^{14},v_7=w^{10},v_8=w^8,v_9=w^4,v_{10}=w^2.
\end{align*} 
Then $v_m^2=u_m$ for $1\leq m\leq 10$. Thus by Corollary \ref{thmsdk1k2}, the code $C_{k-2,k-1,\bm{v}}$ is a self-dual code over $\mathbb{F}_q$ with a generator matrix 
\begin{align*}
    G_{k-2,k-1,\bm{v}}&=G_{k-2,k-1}\cdot \text{diag}(\bm{v})\\
    &=\begin{bmatrix}
        w^{13} & w^{11}  & w^7 & w^5 &   w& w^{14} & w^{10} & w^8 & w^4 & w^2\\
w^{14} & w^{13} & w^{11} & w^{10} & w^8 &  w^7 & w^5 &  w^4  & w^2  &  w\\
   1 &   1   & 1    &1   & 1   & 1    &1  &  1   & 1  &  1\\
 w^3&  w^6 &w^{12}&    1 & w^6 & w^9 &   1 & w^3  &w^9& w^{12}\\
w^4 & w^8 &   w & w^5 &w^{13}&  w^2& w^{10}& w^{14}&  w^7 &w^{11}
    \end{bmatrix}.
\end{align*}
\end{example}

\begin{example}
 Let $n=10$, $k=5$, $q=5^2$ and $\mathbb{F}_q=\mathbb{F}_5[w]$ with $w^2+4w+2=0$. For $\Lambda=\{w, w^4, w^5, 2, w^8, w^{13}, w^{16}, w^{17}, 3, w^{20}\}\subseteq\mathbb{F}_q$, we have $S_1=S_2=S_3=0$ and 
 \begin{align*}
     u_1=3, u_2=3,u_3=3,u_4=2,u_5=2,u_6=2,u_7=2,u_8=2,u_9=3,u_{10}=3.
 \end{align*}
Take \begin{align*}
    v_1=w^9,  v_2=w^9,v_3=w^9,v_4=w^3,v_5=w^3,v_6=w^3,v_7=w^3,v_8=w^3,v_9=w^9,v_{10}=w^9.
\end{align*} 
Then $v_m^2=u_m$ for $1\leq m\leq 10$. Thus by Corollary \ref{thmsdk1k2}, the code $C_{k-2,k-1,\bm{v}}$ is a self-dual code over $\mathbb{F}_q$ with a generator matrix 
\begin{align*}
    G_{k-2,k-1,\bm{v}}&=G_{k-2,k-1}\cdot \text{diag}(\bm{v})\\
    &=\begin{bmatrix}
         w^9 & w^9 & w^9 & w^3 & w^3 & w^3 & w^3 & w^3  &w^9&  w^9\\
w^{10}& w^{13}& w^{14}&  w^9& w^{11}& w^{16}& w^{19}& w^{20}&  w^3&  w^5\\
w^{11}& w^{17}& w^{19}& w^{15}& w^{19}&  w^5& w^{11}& w^{13}& w^{21}&    w\\
w^{14}&  w^5& w^{10}  &w^9 &w^{19}& w^{20} &w^{11} &w^{16}  &w^3& w^{13}\\
w^{15}  &w^9 &w^{15} &w^{15}&  w^3&  w^9&  w^3&  w^9& w^{21} & w^9
    \end{bmatrix}.
\end{align*}
\end{example}

\begin{remark}\label{k34}
For $k=3,4$, it is easy to see that the conditions in Theorem \ref{thmsok1k2} and Corollary \ref{thmsdk1k2} are sufficient conditions for self-orthogonality and self-duality, respectively.  
\end{remark}

\begin{example}
Let $q=23$, $n=9$, $k=3$, and $\{\alpha_1,\alpha_2,\dots,\alpha_9\}=\{0,1,2,3,4,5,6,7,18\}$. Then, we have 
$$u_1=3,u_2=17,u_3=22,u_4=12,u_5=20,u_6=3,u_7=20,u_8=16,u_9=2.$$  
Take  $f(x)=x^3+21x+18\in \mathbb{F}_q[x]$. Then
\begin{align*}
    \begin{split}
        u_1f(\alpha_1)=8, u_2f(\alpha_2)=13,u_3f(\alpha_3)=1,u_4f(\alpha_4)=8,u_5f(\alpha_5)=8,\\
         u_6f(\alpha_6)=8, u_7f(\alpha_7)=1,u_8f(\alpha_8)=9,u_9f(\alpha_9)=13.
    \end{split}
\end{align*}
Observe that $u_mf(\alpha_m)$'s, $1\leq m\leq n$, are squares. Therefore, we can take $v_i$'s such that $v_m^2=u_mf(\alpha_m)$ for all $1\leq m\leq n$, that are, 
\begin{align*}
    v_1=13,v_2= 6,v_3= 1, v_4=13, v_5=13, v_6=13, v_=1, v_8=3, v_9=6.
\end{align*}
Also, $$S_1=0, S_2=2, \text{ and } S_3=5.$$
Consequently,
\begin{align*}
    f_2+f_3S_1=0, f_1+f_2S_1+f_3S_2=0, \text{ and } f_0+f_1S_1+f_2S_2+f_3S_3=0.
\end{align*}
Hence, by Remark \ref{k34}, the code $C_{1,2,\bm{v}}$ is self-orthogonal with a generator matrix
\begin{align*}
    G_{2,1,\bm{v}}=\begin{bmatrix}
        13 &  6 & 1& 13& 13& 13 & 1  &3&  6\\
 0 & 6  &8 & 6&  4 &15&  9 &17&  9\\
 0 & 6& 16 &18& 16 & 6&  8&  4 & 1
    \end{bmatrix}.
\end{align*} 

% with parameters $[9,3,6]$. 
\end{example}

\begin{example}
    Let $q=17$, $n=8$, and $k=4$. For $$\{\alpha_1,\alpha_2,\dots,\alpha_n\}=\{3,5,6,7,10,11,12,14\},$$ we get $S_1=S_2=S_3=0$ and $$u_1=6,u_2=10,u_3=12,u_4=14,u_5=3,u_6=5,u_7=7,u_8=11.$$ Take $v=(1, 9, 11, 12, 3, 7, 15, 4)$. Then $v_m^2=3u_m$, $1\leq m\leq 8$. By Remark \ref{k34}, the code $C_{2,3,\bm{v}}$ is a self-dual code  with a generator matrix
    \begin{align*}
        G_{2,3,\bm{v}}=\begin{bmatrix}
             1 & 9 &11 &12  &3  &7 &15 & 4\\
 3& 11& 15 &16& 13&  9& 10&  5\\
13& 15& 10 &14& 12& 11&  8&  1\\
 5 & 7 & 9& 13 & 1&  2& 11& 14
        \end{bmatrix}.
    \end{align*} 
\end{example}

\section{The second family of linear codes}\label{secondfam}

Let $n,k$, and $h$ be positive integers with $3\leq k\leq h \leq n-1$, define  $$G_{h,k}=\begin{bmatrix}
    1&1&\dots&1\\
    \alpha_1&\alpha_2&\dots&\alpha_n\\
    \vdots&\vdots&\vdots&\vdots\\
    \alpha_1^{k-2}&\alpha_2^{k-2}&\dots&\alpha_n^{k-2}\\
    \alpha_1^h&\alpha_2^h&\dots&\alpha_n^h
\end{bmatrix}.$$
Let $C_{h,k}$ be the linear code generated by $G_{h,k}$. It is easy to see that $C_{h,k}$ is an $[n,k]$ code. 
Note that if $q=2^m$ for some positive integer $m$ and $k=3$, then the MDS code $C_{h,3}$ of length $q$ corresponds to a planar arc, which can be extended to a hyperoval in $PG(2,q)$, and the monomial $x^h$ is called an $o$-monomial (for details,  see \cite{Ball2019,Hirschfeld1998}). In this section, we establish necessary and sufficient conditions on the set $\Lambda=\{\alpha_1,\alpha_2,\dots,\alpha_n\}$ such that code $C_{h,k}$ is MDS. Furthermore, we study the non-GRS property of $C_{h,k}$. First, we determine a parity check matrix.

\subsection{Parity check matrix for code $C_{h,k}$}
In this subsection, we determine a parity check matrix for code $C_{h,k}$.

\begin{theorem}
Let $C_{h,k}$ be the code generated by $G_{h,k}$. Let $h=(k-1)+r$, $1\leq r \leq n-k$. Then
$$H_{h,k}=\begin{bmatrix}
    u_1&u_2&\dots&u_n\\
    u_1\alpha_1&u_2\alpha_2&\dots&u_n\alpha_n\\
    \vdots&\vdots&\vdots&\vdots\\
    u_1\alpha_1^{n-k-r-1}&u_2\alpha_2^{n-k-r-1}&\dots&u_n\alpha_n^{n-k-r-1}\\
    u_1\beta_{11}&u_2\beta_{21}&\dots&u_n\beta_{n1}\\
      u_1\beta_{12}&u_2\beta_{22}&\dots&u_n\beta_{n2}\\
        \vdots&\vdots&\vdots&\vdots\\
         u_1\beta_{1r}&u_2\beta_{2r}&\dots&u_n\beta_{nr}
\end{bmatrix}$$
is a parity check matrix for $C_{h,k}$, where
$$\beta_{ia}=\sum_{j=0}^a(-1)^j\sigma_j\alpha_i^{n-k-r+(a-j)}$$
for $1\leq i\leq n$ and $1\leq a\leq r$.
\end{theorem}
\begin{proof}
Using Lemma \ref{paritylemma}, it is easy to see that all rows of $G_{h,k}$ are orthogonal to the first $(n-k-r)$ rows of $H_{h,k}$. Let $g_l$ and $\bar{h}_m$ denote $l^{\mathrm{th}}$ and $m^{\mathrm{th}}$ rows ($1\leq m\leq n-k-r$) of $G_{h,k}$ and $H_{h,k}$, respectively. For $1\leq a\leq r$, we denote the $(n-k-r+a)$-th row of $H_{h,k}$ by   $h_a$. Then, for $1\leq l\leq k-1$, we have 
\begin{align*}
    g_lh_a^T=&\sum_{i=1}^nu_i\alpha_i^{l-1}\beta_{ia}\\
    =&\sum_{i=1}^nu_i\alpha_i^{l-1}\sum_{j=0}^a(-1)^j\sigma_j\alpha_i^{n-k-r+(a-j)}\\
    =&\sum_{j=0}^a(-1)^j\sigma_j\sum_{i=1}^nu_i\alpha_i^{n-k-r+(a-j)+(l-1)}\\
    = & 0
\end{align*}
by Lemma \ref{paritylemma}, since $n-k-r+(a-j)+(l-1)\leq n-2$. 
For $l=k$, we have 
\begin{align*}
    g_lh_a^T=&\sum_{i=1}^nu_i\alpha_i^{h}\beta_{ia}\\
    =&\sum_{i=1}^nu_i\alpha_i^{(k-1)+r}\sum_{j=0}^a(-1)^j\sigma_j\alpha_i^{n-k-r+(a-j)}\\
    =&\sum_{j=0}^a(-1)^j\sigma_j\sum_{i=1}^nu_i\alpha_i^{n-k-r+(a-j)+(k-1+r)}\\ =&\sum_{j=0}^a(-1)^j\sigma_j\sum_{i=1}^nu_i\alpha_i^{n+(a-j)-1}\\
    =&\sum_{j=0}^a(-1)^j\sigma_j S_{a-j}\ \ \\
    =&0 
\end{align*}
by Lemma \ref{paritylemma} and Eq. \ref{eqA}. Therefore, each row of $G_{h,k}$ is orthogonal to every row of $H_{h,k}$. Next, we show that the rows of $H_{h,k}$ are linearly independent. For $1\leq a\leq r$, let $$v_a=\left (\alpha_1^{k-(a-r+1)},\alpha_2^{k-(a-r+1)},\dots,\alpha_n^{k-(a-r+1)}\right ).$$
Then 
\begin{align*}
    v_ah_a^T=&\sum_{i=1}^nu_i\alpha_i^{k-(a-r+1)}\beta_{ia}\\
          =&\sum_{i=1}^nu_i\alpha_i^{k-(a-r+1)}\sum_{j=0}^a(-1)^j\sigma_j\alpha_i^{n-k-r+(a-j)}\\
          =& \sum_{j=0}^a(-1)^j\sigma_j\sum_{i=1}^nu_i\alpha_i^{n-j-1}\\
          =& 1+\sum_{j=1}^a(-1)^j\sigma_j\sum_{i=1}^nu_i\alpha_i^{n-j-1}\\
          =&1
\end{align*}
by Lemma \ref{paritylemma}. Also, for $1\leq a, b\leq r$ and $a \neq b$, we have 
\begin{align*}
    v_ah_b^T=&\sum_{i=1}^nu_i\alpha_i^{k-(a-r+1)}\beta_{ib}\\
          =&\sum_{i=1}^nu_i\alpha_i^{k-(a-r+1)}\sum_{j=0}^b(-1)^j\sigma_j\alpha_i^{n-k-r+(b-j)}\\
          =& \sum_{j=0}^b(-1)^j\sigma_j\sum_{i=1}^nu_i\alpha_i^{n-j-1+(b-a)}\\
          =& \begin{cases}
        \sum_{j=0}^b(-1)^j\sigma_j\sum_{i=1}^nu_i\alpha_i^{n-j-1+(b-a)}=0 & \text{ if } b-a<0\\  \\
        \sum_{j=0}^b(-1)^j\sigma_j\sum_{i=1}^nu_i\alpha_i^{n-j-1+(b-a)} \\= \sum_{j=0}^b(-1)^j\sigma_jS_{(b-a)-j} & \text{ if } b-a>0\\
        = \sum_{j=0}^{b-a}(-1)^j\sigma_jS_{(b-a)-j}+\sum_{j=b-a+1}^b(-1)^j\sigma_jS_{(b-a)-j}\\
        =0 
          \end{cases}          
\end{align*}
by Eq. \ref{eqA} and $S_t=0$ for $  t<0$, that is, $v_ah_b^T=0$ for $1\leq a, b\leq r$ and $a\neq b$. Further, for $1\leq m\leq n-k-r$ and $1\leq a\leq r$, we have
\begin{align*}
    v_a\bar{h}_m^T=&\sum_{i=1}^nu_i\alpha_i^{k-(a-r+1)}\alpha_i^{m-1}\\
          =&\sum_{i=1}^nu_i\alpha_i^{k+r+m-a-2}\\
          =& 0 \ \ \ 
\end{align*}
by Lemma \ref{paritylemma}, since   $k+r+m-a-2\leq n-2$.
Now we show that the rows of $H_{h,k}$ are linearly independent. Assume that rows are linearly dependent. It is easy to see that the first $(n-k-r)$ rows of $H_{h,k}$ are linearly independent. Choose the minimal $a$,  such that $h_a$ is a linear combination of preceding rows $\bar{h}_1,\bar{h}_2,\dots,\bar{h}_{n-k-r},h_1,\dots,h_{a-1}$. Let  $\mathcal{L}=\text{span}\{\bar{h}_1,\bar{h}_2,\dots,\bar{h}_m,h_1,\dots,h_{a-1}\}$. However, $v_a$  is orthogonal to $\mathcal{L}$ but not orthogonal to $h_a$, a contradiction. Hence, the rows are linearly independent. This completes the proof.  
\end{proof}

% Observe that code $C_{h,k}$ is not MDS in general. It depends on the choice of the set $\Lambda\subseteq\mathbb{F}_q$, as illustrated in the following example.
% \begin{example}
%  Let $q=11$, $k=3$, $h=6$, $n=8$ and $\Lambda=\{1,2,3,4,5,6,7,8\}$. Then the code generated by the matrix
%  $$G_{6,3}=\begin{bmatrix}
%      1&1&1&1&1&1&1&1\\
%      1&2&3&4&5&6&7&8\\
%      1&2^6&3^6&4^6&5^6&6^6&7^6&8^6\\
%      \end{bmatrix}$$
%  has parameters $[8,3,4]$ which is not MDS.    
% \end{example}

\subsection{MDS and non-GRS properties of code $C_{h,k}$}
In this subsection, we determine the necessary and sufficient conditions under which the code $C_{h,k}$ attains the MDS property and examine its non-GRS properties. We further present explicit constructions and examples. In the next theorem, we provide the necessary and sufficient condition for $C_{h,k}$ to be MDS. 
\begin{theorem}\label{thmmdshk}
 The code $C_{h,k}$ generated by $G_{h,k}$
is an MDS code if and only if $$S_{h-k+1}(\beta_1,\beta_2,\dots,\beta_k)\neq0$$ for all $\{\beta_1,\beta_2,\dots,\beta_k\}\subseteq \Lambda$ with distinct $\beta_i$.
\end{theorem}
\begin{proof}
By Lemma \ref{mdsnonsingular}, $C_{h,k}$ is MDS if and only if each $k\times k$ submatrix of $G_{h,k}$ is non-singular. Let $B$ be an arbitrary $k\times k$ submatrix of $G_{h,k}$. Let 
  $$B=\begin{bmatrix}
    1&1&\dots&1\\
    \beta_{1}&\beta_{2}&\dots&\beta_{_k}\\
    \vdots&\vdots&\vdots&\vdots\\
    \beta_{1}^{k-2}&\beta_{2}^{k-2}&\dots&\beta_{k}^{k-2}\\
    \beta_{1}^h&\beta_{2}^h&\dots&\beta_{k}^h
\end{bmatrix}$$ for some $\{\beta_1,\beta_2,\dots,\beta_k\}\subseteq \Lambda$. Then by Lemma  \ref{detgh}, $\det(B)\neq 0$ if and only if $S_{h-k+1}(\beta_1,\beta_2,\dots,\beta_k)\neq0$. This completes the proof.
\end{proof}

Next, we show that the family of codes $C_{h,k}$ constitutes a class of MDS codes that are not generalized Reed–Solomon codes. The proof of their non-GRS nature relies on the Schur product method based on Lemma \ref{lemmanongrs}. In \cite{Han2024}, the authors studied MDS and NMDS code for $h=k$, but not investigated non-GRS properties.
\begin{proposition}   
Let $3\leq k\leq \frac{n-2}{2}$, $h=k$ (that is $r=1)$ be positive integers, and $C_{h,k}$ be an MDS code. Then $C_{h,k}$ is a non-GRS MDS code.   
\end{proposition}
\begin{proof}
 Denote $\Lambda^i=(\alpha_1^i,\alpha_2^i,\dots,\alpha_n^i)$. Then, the Schur square of $C_{h,k}$ is generated by 
$$T=\{\Lambda^i: 0\leq i \leq 2k-2\}\cup\{\Lambda^{2k}\}.$$
Thus, $\dim(C_{h,k}\star C_{h,k})=2k$. Hence, by Lemma \ref{lemmanongrs}, $C_{h,k}$ is non-GRS.     
\end{proof}

\begin{theorem}\label{thmnongrshk1}
Let $4\leq k\leq \frac{n-1}{2}$, $h=(k-1)+r$  with $2\leq r\leq k-2$, and $C_{h,k}$ be an MDS code.  Then $C_{h,k}$ is a non-GRS MDS code.   
\end{theorem}
\begin{proof}
Denote $\Lambda^i=(\alpha_1^i,\alpha_2^i,\dots,\alpha_n^i)$.  
Since $2\leq r\leq k-2$,  the Schur square of $C_{h,k}$ is generated by 
$$T=\{\Lambda^i: 0\leq i \leq 2k-3+r\}\cup\{\Lambda^{2k-2+2r}\}.$$
Observe that $\{\Lambda^i: 0\leq i \leq 2k-1\}\subseteq T$ is linearly independent as $2k-1\leq q-2$, that is, $\dim(C_{h,k}\star C_{h,k})\geq 2k$. Hence, by Lemma \ref{lemmanongrs}, $C_{h,k}$ is non-GRS.    
\end{proof}

  \begin{theorem}\label{thmnongrshk2}
(1) Let $k=3\leq \frac{n}{2}$, $h=(k-1)+r$ with $k-1\leq r\leq \frac{q-k-3}{2}$ and $C_{h,k}$ be an MDS code. Then $C_{h,k}$ is a non-GRS MDS code.\\
(2) Let $4\leq k\leq \frac{n}{2}$, $h=(k-1)+r$  with $k-1 \le r \le q-k-3$,  and $C_{h,k}$ be an MDS code.  Then $C_{h,k}$ is a non-GRS MDS code.   
\end{theorem}  
\begin{proof}
% Since $2k\leq n-4\leq q-4$, therefore $k-1\leq q-k-3$, that is, $\{k-1,k,\dots,q-k-3\}\neq \emptyset$.
 We denote $\Lambda^i=(\alpha_1^i,\alpha_2^i,\dots,\alpha_n^i)$.
(1) For $k=3$, the Schur square of $C_{h,k}$ is generated by 
\begin{align*}
 T= \{\Lambda^0, \Lambda^1, \Lambda^2, \Lambda^h, \Lambda^{h+1},\Lambda^{2h}\}.  
\end{align*}
Moreover, since $2h=q-k+1\leq q-2$, the vectors in $T$ are linearly independent.  Thus, $\dim(C_{h,k}\star C_{h,k})=6=2k$. Hence, by Lemma \ref{lemmanongrs}, $C_{h,k}$ is non-GRS. \\ 
(2) For $k\geq 4$,  the Schur square of $C_{h,k}$  is generated by 
  $$T=T_1\cup T_2 \cup \{\Lambda^{2k-2+2r}\},$$
where $T_1=\{\Lambda^i: 0\leq i \leq 2k-4\}$ and $T_2=\{\Lambda^{l}: (k-1)+r\leq l\leq 2k-3+r\}$. Observe that the vectors in $T_1$ are linearly independent as $2k-4< q-2$. Further, $\{\Lambda^{k-1+r},\Lambda^{k+r},\Lambda^{k+r+1}\}\subseteq T_2$ is linearly independent as $k+r+1\leq q-2$. Thus, $\dim(C_{h,k}\star C_{h,k})\geq (2k-3)+3=2k$. Hence, by Lemma \ref{lemmanongrs}, $C_{h,k}$ is non-GRS.  
\end{proof}

\begin{remark}\label{notequivalent}
It is important to note that if two $[n,k]$ codes $C_1$ and $C_2$ are equivalent, then their Schur squares of $C_1$ and $C_2$ are also equivalent.  For $5\leq k\leq \frac{n-2}{2}$ with $n\leq q$, let $C_{h,k}$ be the code generated by $G_{h,k}$, where $h=(k-1)+r$  with $2\leq r\leq k-2$. Analogous to Theorem \ref{thmnongrshk1},  the dimension of the Schur square of  $C_{h,k}$ is at least $2k+1$. By contrast, in \cite{Liu2025, Li2024}, the authors constructed non-GRS MDS codes whose Schur squares have dimension $2k$ over the same range. Hence, the class of codes $C_{h,k}$ contains non-GRS MDS codes that are not equivalent to codes presented in \cite{Liu2025,Li2024}.      
\end{remark}

Using techniques similar to those in  \cite{Jin2024}, we now present an explicit construction of non-GRS MDS codes.  
\begin{theorem}\label{conschk}
Let $4\leq k\leq \frac{n-1}{2}$, $h=k+1$, $q=p^m$, $\mathbb{F}_q=\mathbb{F}_p[\gamma]$ and $t=\lfloor \frac{m-1}{2}\rfloor$. Let $n\leq p^t$ and $\Lambda=\{\alpha_i=\gamma^t+\sum_{j=0}^{t-1} a_{ij}\gamma^j: 1\leq i\leq n, a_{i,j}\in \mathbb{F}_p\}$. If $p\nmid \frac{k(k+1)}{2}$, then the code $C_{h,k}$ is a non-GRS MDS code over $\mathbb{F}_q$.    
\end{theorem}
\begin{proof}
 Let $(\beta_1,\beta_2,\dots,\beta_k)=(\alpha_{i_1},\alpha_{i_2},\dots,\alpha_{i_k})$ for $\{i_1,i_2,\dots,i_k\}\subset\{1,2,\dots,n\}$. Then
 \begin{align*}
     S_{h-k+1}(\beta_1,\beta_2,\dots,\beta_k)=&S_2(\beta_1,\beta_2,\dots,\beta_k)=\sigma_1^2(\beta_1,\beta_2,\dots,\beta_k)-\sigma_2(\beta_1,\beta_2,\dots,\beta_k)\\
     =&\left (\sum_{1\leq j_1\leq k}\alpha_{i_{j_1}}\right)^2-\sum_{1\leq j_1<j_2\leq k}\alpha_{i_{j_1}}\alpha_{i_{j_2}}\\
     =& \left(k^2-\binom{k}{2}\right)\gamma^{2t}+\sum_{j=0}^{2t-1}c_{j}\gamma^j \ \  \text{ for some } c_{j}\in \mathbb{F}_p\\
     =& \frac{k(k+1)}{2}\gamma^{2t}+\sum_{j=0}^{2t-1}c_{j}\gamma^j\neq 0.
 \end{align*}
 Hence, by Theorem \ref{thmmdshk} and Theorem \ref{thmnongrshk1}, $C_{h,k}$ is a non-GRS MDS code.
\end{proof}

Next, we present some examples of non-GRS MDS codes with larger length. 

\begin{example}
 Let $q=37,n=18, h=21$, $k=4$. For $\Lambda=\{3^i: 1\leq i\leq n\}$,  the code $C_{21,4}$ generated by $G_{21,4}$ is MDS by Theorem \ref{thmmdshk}. Also, $h=(k-1)+r$, where $r=18\leq q-k-3=30$. Hence, by Theorem \ref{thmnongrshk2}, $C_{21,4}$ is non-GRS MDS. Moreover, we get a non-GRS MDS code of length $n+1$ over $\mathbb{F}_{q}$ by adding a column  $(0,0,1)^T$ in  $G_{21,4}$. 
\end{example}

\begin{example}\label{example1}
    Let $n=q=2^7$, $h=16$, $k=3$ and  $\Lambda=\mathbb{F}_q$.
Let $C_{16,3}$ be the code generated by 
\begin{equation*}\label{opolymatrix}
   G_{16,3}= \begin{bmatrix}
    1&1&1&\dots &1\\
    \alpha_1&\alpha_2&\alpha_3&\dots&\alpha_q\\
\alpha_1^h&\alpha_2^h&\alpha_3^h&\dots&\alpha_q^h
\end{bmatrix}.
\end{equation*}
Then $S_{14}(\beta_1,\beta_2,\beta_3)\neq 0$ for all  distinct $\{\beta_1,\beta_2,\beta_3\}\subset\mathbb{F}_q$. Thus, by Theorems \ref{thmmdshk} and \ref{thmnongrshk2},  $C_{16,3}$ is a non-GRS MDS code of length $q$. Moreover, we get an MDS code of length $q+2$ over $\mathbb{F}_{q}$ by adding two columns $(0,1,0)^T$ and $(0,0,1)^T$ in  $C_{16,3}$. The corresponding $o$-monomial $x^{16}$ is a  translation $o$-monomial (see \cite{Ball2019}).
\end{example}
In \cite{Han2024}, the authors studied MDS and NMDS codes by taking $h=k$, that is, $r=1$. In Table \ref{tab1}, we construct several non-GRS MDS codes over $\mathbb{F}_q$ of length $(q-1)/2$, using Theorem \ref{thmnongrshk2} by taking $r>k-2$.
\begin{table}[h]
    \centering
    \begin{tabular}{|c|c|c|c|c|c|}
    \hline
      \hspace{0.5 cm}  $q$ &\hspace{0.5 cm}  $n$ &\hspace{0.5 cm} $k$ &\hspace{0.5 cm} $r$ &\hspace{0.5 cm} $\gamma$ &\hspace{0.5 cm} \textbf{Parameters}\\
      \hline
          $37$ & $18$  & $7$   & $18$ & $3$ & $[18,7,12]$\\
          \hline
      $41$ & $20$  & $8$   & $20$ & $2$ & $[20,8,13]$\\
      \hline
          $53$ & $26$  & $11$   & $26$ & $4$ & $[26,11,16]$\\
          \hline
          $61$ & $30$  & $13$   & $30$ & $4$ & $[30,13,18]$\\
          \hline
          $73$ & $36$  & $16$   & $36$ & $6$ & $[36,16,21]$\\
          \hline
          $89$ & $44$  & $20$   & $44$ & $5$ & $[44,20,25]$\\
          \hline
      
    \end{tabular}
    \caption{Non-GRS MDS codes generated by $G_{h,k}$, using Theorem \ref{thmnongrshk2}, where $h=(k-1)+r$ and $\Lambda=\{\gamma^i: 1\leq i\leq n\}$.}
    \label{tab1}
\end{table}

\begin{remark}
Theorem \ref{conschk} shows that non-GRS MDS codes up to length $\left (\frac{q}{p}\right )^{1/2}$ can be constructed using the generator matrices $G_{h,k}$. Moreover, Example \ref{example1} shows that non-GRS MDS codes of length $q+2$ over $\mathbb{F}_{q}$, for $q=2^m$ and positive integer $m$, can be constructed using Theorem \ref{thmmdshk}. In addition, the examples in Table \ref{tab1} demonstrate that the matrix $G_{h,k}$ can be utilized to construct non-GRS MDS codes of length $(q-1)/2$ over $\mathbb{F}_q$ ($q$ is odd). 
\end{remark}

\subsection{Self-orthogonal and self-dual properties of code $C_{h,k}$ }\label{orthogonalproperties}
In this subsection, we investigate the self-orthogonality and self-duality of the codes $C_{h,k}$. 

Let  ${\bm{v}}=(v_1,v_2,\dots,v_n)\in (\mathbb{F}_q^*)^{n}$. Consider a matrix $$G_{h,{\bm{v}}}=G_{h,k}\cdot \text{diag}[v_1,v_2,\dots,v_n].$$
We denote the linear code generated by $G_{h,{\bm{v}}}$ by $C_{h,{\bm{v}}}$. The codes $C_{h,\bm{v}}$ and $C_{h,k}$ are equivalent.

\begin{theorem}\label{thmsoh1}
Let $k\geq 3$, $h=(k-1)+r$ with $1\leq r\leq k-2$,  $\bm{v}=(v_1,v_2,\dots,v_n)\in (\mathbb{F}_q^*)^{n}$, and $u_i$, $1\leq i\leq n$ be as in Eq. \ref{ui}. Then $C_{h,\bm{v}}$ is self-orthogonal if and only if there exists a polynomial $f(x)=\sum_{j=0}^{n-2k-r+1}f_jx^j\in \mathbb{F}_q[x]$  such that the following hold: 
  \begin{enumerate}
      \item $v_i^2=u_if(\alpha_i)$ for $1\leq i \leq n$.
      \item $\sum_{j=s}^{s+r}f_jS_{j-s}=0$, where $s=n-2k-2r+1$ and $f_j=0$ if $j<0$.
  \end{enumerate}
\end{theorem}
\begin{proof}
    By Lemma \ref{selflemma}, $C_{h,\bm{v}}$ is self-orthogonal if and only if $G_{h,\bm{v}}G_{h,\bm{v}}^T=0$, that is,
    \begin{align}
          &\sum_{i=1}^nv_i^2\alpha_i^l=0  \text{ for all } l \text{ with } 0\leq l\leq 2k-3+r, \label{eq1hv}\\
            &\sum_{i=1}^nv_i^2\alpha_i^{2k-2+2r}=0.\label{eq2hv}    
   \end{align}
   Let $C_{h,\bm{v}}$ be self-orthogonal, that is, Eqs. \ref{eq1hv}, \ref{eq2hv} hold. Then, Eq. \ref{eq1hv} implies that  $(v_1^2,v_2^2,\dots,v_n^2)$ is a nonzero solution of the following system of linear equations
 \begin{align} \label{eq3hv}
           &\sum_{i=1}^n\alpha_i^l x_i=0  \text{ for } 0\leq l\leq 2k-3+r. 
        \end{align}
The rank of the coefficient matrix of Eq. \ref{eq3hv} is $2k-2+r$. Therefore, the dimension of the solution space of Eq. \ref{eq3hv} is $n-2k-r+2$. Also, by Lemma \ref{paritylemma},
\begin{align*}
\begin{split}
    (u_1,u_2,\dots,u_n), (u_1\alpha_1,u_2\alpha_2,\dots,u_n\alpha_n),\dots,(u_1\alpha_1^{n-2k-r+1},u_2\alpha_2^{n-2k-r+1},\dots,u_n\alpha_n^{n-2k-r+1})
    \end{split}
    \end{align*}
   are $n-2k-r+2$ linearly independent solutions of Eq. \ref{eq3hv}. This implies that there exist scalars $f_0,f_1,\dots,f_{n-2k-r+1}$ such that 
   \begin{align*}
   \begin{split}
       v_i^2=\sum_{j=0}^{n-2k-r+1} u_if_j\alpha_i^j.
        \end{split}
   \end{align*}
Denote $f(x)=\sum_{j=0}^{n-2k-r+1}f_jx^j$. Then  $v_i^2= u_if(\alpha_i)$ for $1\leq i\leq n$. Furthermore, $v_i^2$, $1\leq i\leq n$, also satisfies Eq. \ref{eq2hv}, therefore, we have
\begin{align*}
0=\sum_{i=1}^nu_if(\alpha_i)\alpha_i^{2k-2+2r}=\sum_{j=0}^{n-2k-r+1}f_j\sum_{i=1}^nu_i\alpha_i^{2k-2+2r+j}=\sum_{j=0}^{n-2k-r+1}f_jS_{2k-2+2r+j-n+1}\\
=\sum_{j=s}^{s+r}f_jS_{j-s}=0, \text{ where } s=n-2k-2r+1,
\end{align*}
using Lemma \ref{paritylemma}. This completes the necessary part. The proof of the converse part is similar to that of Theorem \ref{thmsok1k2}.
\end{proof}

\begin{theorem}\label{thmsoh2}
 Let $k\geq 3$, $h=(k-1)+r$ with $k-1\leq r\leq q-k-1$, $\bm{v}=(v_1,v_2,\dots,v_n)\in (\mathbb{F}_q^*)^{n}$, and $u_i$, $1\leq i\leq n$ be as in Eq. \ref{ui}. Then $C_{h,\bm{v}}$ is self-orthogonal if and only if there exists a polynomial $f(x)=\sum_{j=0}^{n-2k+2}f_jx^j\in \mathbb{F}_q[x]$ such that the following hold  
  \begin{enumerate}
      \item $v_i^2=u_if(\alpha_i)$ for $1\leq i \leq n$.
      \item $\sum_{j=b}^{b+2r+1}f_jS_{j-b}=0$, where $b=n-2k-2r+1$.
      \item $\sum_{j=s}^{s+l-2k+3}f_jS_{j-s}=0$  for  $l$ with $(k-1)+r\leq l\leq (2k-3)+r$,  where $s=n-l-1$ and $f_j=0$ if $j<0$.    
  \end{enumerate}
\end{theorem}
\begin{proof}
    By Lemma \ref{selflemma}, $C_{h,\bm{v}}$ is self-orthogonal if and only if $G_{h,\bm{v}}G_{h,\bm{v}}^T=0$, that is,
    \begin{align}
          &\sum_{i=1}^nv_i^2\alpha_i^l=0  \text{ for } 0\leq l\leq 2k-4, \label{eq4.4}\\
           &\sum_{i=1}^nv_i^2\alpha_i^l=0  \text{ for all } l \text{ with } (k-1)+r\leq l\leq (2k-3)+r,\label{eq4.5} \\
            &\sum_{i=1}^nv_i^2\alpha_i^{2k-2+2r}=0\label{eq4.6}.&    
   \end{align}
Let $C_{h,\bm{v}}$ be self-orthogonal, that is, Eqs. \ref{eq4.4}, \ref{eq4.5}, \ref{eq4.6} hold. Then, Eq. \ref{eq4.4} implies that  $(v_1^2,v_2^2,\dots,v_n^2)$ is a nonzero solution of the following system of linear equations
 \begin{align} \label{eq4.7}
           &\sum_{i=1}^n\alpha_i^l x_i=0  \text{ for all } l \text{ with } 0\leq l\leq 2k-4. 
        \end{align}
The rank of the coefficient matrix of Eq. \ref{eq4.7} is $2k-3$. Therefore, the dimension of the solution space of Eq. \ref{eq4.7} is $n-2k+3$. Also, by Lemma \ref{paritylemma},
\begin{align*}
\begin{split}
    (u_1,u_2,\dots,u_n), (u_1\alpha_1,u_2\alpha_2,\dots,u_n\alpha_n),\dots,(u_1\alpha_1^{n-2k+2},u_2\alpha_2^{n-2k+2},\dots,u_n\alpha_n^{n-2k+2})
    \end{split}
    \end{align*}
   are $n-2k+3$ linearly independent solutions of Eq. \ref{eq4.7}. This implies that there exist scalars $f_0,f_1,\dots,f_{n-2k+2}$ such that 
   \begin{align*}
   \begin{split}
       v_i^2=\sum_{j=0}^{n-2k+2} u_if_j\alpha_i^j.
        \end{split}
   \end{align*}
Denote $f(x)=\sum_{j=0}^{n-2k+2}f_jx^j$. Then  $v_i^2= u_if(\alpha_i)$ for $1\leq i\leq n$. Furthermore, $v_i^2$, $1\leq i\leq n$, also satisfies Eqs. \ref{eq4.6}, therefore, we have 
\begin{align*}
    0=\sum_{j=0}^{n-2k+2} f_j\sum_{i=1}^n u_i\alpha_{i}^{2k-2+2r+j}=\sum_{j=0}^{n-2k+2}f_jS_{2k-2+2r+j-n+1}\\
    =\sum_{j=b}^{b+2r+1}f_jS_{j-b}=0, \text{ where } b=n-2k-2r+1,
\end{align*} 
using Lemma \ref{paritylemma}. Moreover, $v_i^2$'s satisfy Eq. \ref{eq4.5} as well, for $(k-1)+r\leq l\leq (2k-3)+r$, 
\begin{align*}
    0=\sum_{j=0}^{n-2k+2}f_j\sum_{i=1}^nu_i\alpha_i^{l+j}=\sum_{j=0}^{n-2k+2}f_jS_{l+j-n+1}\\
    =\sum_{j=s}^{s+l-2k+3}f_jS_{j-s}=0,  \text{ where } s=n-l-1,   \end{align*}
    using Lemma \ref{paritylemma}.
   The converse follows similar to that of Theorem \ref{thmsok1k2}.
\end{proof}

\begin{example}
 Let $q=19$, $n=10$, $k=4$, $r=2$, $\Lambda=\{0,1,2,3,4,5,8,11,15,16\}$ and $f(x)=x+2\in \mathbb{F}_q[x]$. Then $S_1=8$, $S_2=3$ and 
  $$u_1=2,u_2=12,u_3=6,u_4=1,u_5=11,u_6=4,u_7=13,u_8=10,u_9=7,u_{10}=10.$$
  Consequently,
  $$\{u_if(\alpha_i): 1\leq i\leq 10\}=\{4,17,5,5,9,9,16,16,5,9\},$$
  all are squares in $\mathbb{F}_q$. Thus, we take $v_i$'s such that $v_i^2=u_if(\alpha_i)$, $1\leq i\leq 10$, that are, 
  \begin{align*}
   v_1=17, v_2=6,v_3=9,v_4=9,v_5=16,v_6=16,v_7=4,v_8=4,v_9=9,
v_{10}=16.   
  \end{align*}  
  The first condition of Theorem \ref{thmsoh1} is satisfied.
  For the  second condition, note that $s=-1$ and  
  $$f_{-1}S_0+f_0S_1+f_1S_2=0+16+3=0.$$
  Thus, the code $C_v$ is a self-orthogonal code over $\mathbb{F}_q$ with a generator matrix 
  \begin{align*}
      G_{\bm{v}}=\begin{bmatrix}
          17 &  6 & 9  &9 &16 &16 & 4 & 4 & 9& 16\\
 0&  6& 18&  8&  7&  4 &13 & 6&  2 & 9\\
 0 & 6& 17&  5 & 9 & 1 & 9&  9& 11 &11\\
 0 & 6 & 3&  2 & 6& 11& 10&  9 &18 & 7
      \end{bmatrix}.
  \end{align*}
\end{example}

\begin{theorem}\label{thmsdh1}
   Let $n=2k\geq 6$, $\bm{v}=(v_1,v_2,\dots,v_n)\in (\mathbb{F}_q^*)^{n}$, $u_i$, $1\leq i\leq n$ be as in Eq. \ref{ui}, and $h=(k-1)+r$ for some $1\leq r\leq k-2$. Then $C_{h,\bm{v}}$ is self-dual if and only if $r=1$ and  the followings hold 
   \begin{align*}
        (i)& \text{ all } u_i\text{'s }, 1\leq i\leq n, \text{ are squares or non-squares simultaneously}.  \\
       (ii)&\  S_1=0.
 \end{align*}
\end{theorem}
\begin{proof}
    By Theorem \ref{thmsoh1}, $C_{h,\bm{v}}$ is self-dual if and only if there exists a polynomial $f(x)=\sum_{j=0}^{n-2k-r+1}f_jx^j\in \mathbb{F}_q[x]$, that is, $n-2k-r+1\geq0$. Consequently, $r=1$. By Theorem \ref{thmsoh1}, if $r=1$, then $f(x)=f_0$ and $s=-1$. The first condition of Theorem \ref{thmsoh1} implies that $v_i^2=f_0u_i$ for $1\leq i\leq n$, that is,  all $u_i$'s are squares or non-squares simultaneously and $f_0\neq 0$.
By the second condition in Theorem \ref{thmsoh1}, we get 
$$f_{-1}S_0+f_0S_1=0\implies f_0S_1=0\implies S_1=0.$$
This completes the proof. 
 % By Theorem \ref{thmsoh1}, if $r=0$, then $f(x)=f_0+f_1x$ and $s=1$. The second condition in Theorem \ref{thmsoh1} implies that $f_1S_0=0$, that is, $f_1=0$, since $S_0=1$. Combining this with the first condition of Theorem \ref{thmsoh1}, we get $v_i^2=f_0u_i$ for all $1\leq i\leq n$.
\end{proof}

Since every element of a field with even characteristic is a square, we have a simplified characterization for self-duality of the code $C_{h,\bm{v}}$ over $\mathbb{F}_q$ when $q$ is even.

\begin{corollary}\label{qeven}
  Let $q$ be even, $n=2k\geq 6$, $\bm{v}=(v_1,v_2,\dots,v_n)\in (\mathbb{F}_q^*)^{n}$, $u_i$, $1\leq i\leq n$ be as in Eq. \ref{ui}, and $h=(k-1)+r$ for some $1\leq r\leq k-2$. Then $C_{h,\bm{v}}$ is self-dual if and only if $r=1$ and   $S_1=0$.
\end{corollary}
\begin{proof}
    The proof follows from Theorem \ref{thmsdh1}.
\end{proof}

\begin{corollary}
 Let $q\geq 8$ be even, $n=q$, $h=\frac{q}{2},k=\frac{q}{2}$, and $\Lambda=\{\alpha_1,\alpha_2,\dots,\alpha_{q}\}=\mathbb{F}_q$. Then the code $C_{h,\bm{v}}$ generated by $G_{h,\bm{v}}$ is a self-dual code  for some $\bm{v}\in (\mathbb{F}_q^*)^n$ with parameters $[q,\frac{q}{2},\frac{q}{2}]$.  
\end{corollary}
\begin{proof}
 It is obvious that $S_1=0$. By Corollary \ref{qeven}, the code $C_{h,\bm{v}}$ is self-dual. Moreover, as $C_{h,\bm{v}}$ is a subcode of $GRS(n,k+1,\Lambda,v)$ therefore $q/2\leq d(C_{h,\bm{v}})\leq q/2+1$. By Theorem \ref{thmmdshk}, $C_{h,\bm{v}}$ is not MDS. Thus, $d(C_{h,\bm{v}})=q/2$.    
\end{proof}
% \begin{remark}
%  Explicit constructions of self-dual codes over $\mathbb{F}_q$, where $q$ is odd, for the case $r=1$ are given in \cite{Han2024}. 
% \end{remark}

\begin{corollary}
Let $n=2k\geq8$ and $h=(k-1)+r$ with $k-1\leq r\leq q-k-3$. Then code $C_{h,\bm{v}}$ is never self-dual for any $\bm{v}=(v_1,v_2,\dots,v_n)\in (\mathbb{F}_q^*)^{n}$. 
\end{corollary}
\begin{proof}
On the contrary, let $C_{h,\bm{v}}$ be a self-dual code for some $\bm{v}=(v_1,v_2,\dots,v_n)\in (\mathbb{F}_q^*)^{n}$. Then by Lemma \ref{selflemma}, $G_{h,\bm{v}}G_{h,\bm{v}}^T=0$, that is,
    \begin{align*}
          &\sum_{i=1}^nv_i^2\alpha_i^l=0  \text{ for } 0\leq l\leq 2k-4,\\
           &\sum_{i=1}^nv_i^2\alpha_i^l=0  \text{ for all } l \text{ with } (k-1)+r\leq l\leq (2k-3)+r, \\
            &\sum_{i=1}^nv_i^2\alpha_i^{2k-2+2r}=0.
   \end{align*}
This implies that $v_i^2$'s are a nonzero solution of the following system of linear equations  
\begin{align*}
          &\sum_{i=1}^n\alpha_i^lx_i=0  \text{ for } 0\leq l\leq 2k-4,\\
           &\sum_{i=1}^n\alpha_i^lx_i=0  \text{ for all } l \text{ with } (k-1)+r\leq l\leq (2k-3)+r, \\
            &\sum_{i=1}^n\alpha_i^{2k-2+2r}x_i=0.
   \end{align*}
Consider a submatrix of the augmented matrix of the above system of linear equations
\begin{align*}
   A= \begin{bmatrix}
        \alpha_1^0 &\alpha_2^0 &\dots & \alpha_n^0\\
        \alpha_1 &\alpha_2 &\dots & \alpha_n\\
        \vdots &\vdots &\dots &\vdots\\
        \alpha_1^{2k-4} &\alpha_2^{2k-4} &\dots & \alpha_n^{2k-4}\\
        \alpha_1^{(k-1)+r} &\alpha_2^{(k-1)+r} &\dots & \alpha_n^{(k-1)+r}\\
        \alpha_1^{(k-1)+r+1} &\alpha_2^{(k-1)+r+1} &\dots & \alpha_n^{(k-1)+r+1}\\
        \alpha_1^{(k-1)+r+2} &\alpha_2^{(k-1)+r+2} &\dots & \alpha_n^{(k-1)+r+2}\\
    \end{bmatrix}.
\end{align*}
Since $k+r+1\leq k+q-k-3+1=q-2$,  the rank of matrix $A$ is $2k$ (equal to the number of variables). Hence, the above system of linear equations has only a zero solution, which is a contradiction to $\bm{v}=(v_1,v_2,\dots,v_n)\in (\mathbb{F}_q^*)^{n}$.
\end{proof}

% \begin{corollary}
%     Let $n=2k\geq 6$, $\bm{v}=(v_1,v_2,\dots,v_n)\in (\mathbb{F}_q^*)^{n}$, $u_i$, $1\leq i\leq n$ be as in Eq. \ref{ui}, and $h=(k-1)+r$ for some $k-1\leq r\leq q-k-1$. Then $C_{h,\bm{v}}$ is self-dual if and only if there exists a polynomial \textcolor{red}{ $f(x)=f_0+f_1x+f_2x^2$} such that the following hold :\\
%         (i)\ $ v_i^2=u_if(\alpha_i)$. \\
%        (ii)\  $f_0S_{2r-1}+f_1S_{2r}+f_2S_{2r+1}=0.$\\
%        (iii)\ $\sum_{j=s}^{s+l-2k+3}f_jS_{j-s}=0$ for all $l$ with $(k-1)+r\leq l\leq (2k-3)+r$, where $s=2k-l-1$.
% \end{corollary} 
% \begin{proof}
% By putting $n=2k$ in Theorem \ref{thmsoh2}, we get the desired result.
% \end{proof}

\section{Conclusion}\label{conclusion}
In this work, we constructed two families of linear codes by modifying the generator matrix of GRS codes. For these families, we explicitly determined the parity-check matrices and established necessary and sufficient conditions on the evaluation points under which the codes achieve the MDS property. Using the Schur square method, we further showed that these codes exhibit the non-GRS property. As a result, we concluded that the non-GRS MDS codes constructed in this paper contain codes that are not equivalent to the codes in \cite{Liu2025,Li2024}. Further, we investigated self-orthogonality and self-duality of the constructed codes by deriving necessary and sufficient conditions. In addition, we presented several illustrative examples and explicit constructions to demonstrate the results.
The further generalization to the code $C_I$, obtained by deleting an arbitrary number of rows from the Vandermonde matrix, is a natural 
%and mathematically significant 
direction of study. However, this extension is substantially more challenging than the cases $C_{i,j}$. In particular, for such a general code $C_I$, the necessary and sufficient conditions for the MDS property become considerably more complicated, which makes the explicit construction of corresponding MDS codes highly difficult and technically challenging. %In addition, determining an appropriate parity-check matrix for ($C_I$) is itself a non-trivial problem.  
We identify the study of $C_I$ in full generality as an important direction for future research.

% An interesting direction for future work is to examine whether the extended code of the constructed codes remains non-GRS MDS, and whether they are equivalent or inequivalent to known non-GRS MDS codes in the literature.

\section*{Acknowledgment}

The authors were supported by the UAEU-AUA grant G00004614. 

% Gyanendra K. Verma was with
% Department of Mathematical Sciences
% UAE University, UAE, at the time of submission. 

\bibliographystyle{abbrv}
	\bibliography{ref}

@book{cramer1750,
  title={Introduction {\`a} l'analyse des lignes courbes alg{\'e}briques},
  author={Cramer, Gabriel},
  year={1750},
  publisher={Chez les fr{\`e}res Cramer et C. Philibert}
}

@misc{horn2013,
  title={Matrix Analysis},
  author={Horn, R C. and R. Johnson, R.},
  year={2013},
  publisher={Cambridge University Press}
}

@article{Glynn1986,
  title={The non-classical 10-arc of {PG(4, 9)}},
  author={David G. Glynn},
  journal={Discret. Math.},
  year={1986},
  volume={59},
  pages={43-51},
  url={https://api.semanticscholar.org/CorpusID:26149240}
}

@article {Casse1982,
    AUTHOR = {Casse, L. R. A. and Glynn, D. G.},
     TITLE = {The solution to {B}eniamino {S}egre's problem {$I\sb{r,q},$}
              {$r=3,$} {$q=2\sp{h}$}},
   JOURNAL = {Geom. Dedicata},
  FJOURNAL = {Geometriae Dedicata},
    VOLUME = {13},
      YEAR = {1982},
    NUMBER = {2},
     PAGES = {157--163},
      ISSN = {0046-5755},
   MRCLASS = {51E20 (05B25)},
  MRNUMBER = {684151},
MRREVIEWER = {Raymond\ Hill},
       DOI = {10.1007/BF00147659},
       URL = {https://doi-org.uaeu.idm.oclc.org/10.1007/BF00147659},
}

@article{bhagat2025,
  title={Row-Column Twisted {R}eed-{S}olomon codes},
  author={Bhagat, Anuj Kumar and Singh, Harshdeep and Sarma, Ritumoni},
  journal={ArXiv:2509.06919},
  year={2025}
}

@article{Calderbank1995,
  title={Good quantum error-correcting codes exist.},
  author={A. Robert Calderbank and Peter W. Shor},
  journal={Physical review. A, Atomic, molecular, and optical physics},
  year={1995},
  volume={54 2},
  pages={
          1098-1105
        },
  url={https://api.semanticscholar.org/CorpusID:11524969}
}

@article{Shor1995,
  title={Scheme for reducing decoherence in quantum computer memory.},
  author={Peter W. Shor},
  journal={Physical review. A, Atomic, molecular, and optical physics},
  year={1995},
  volume={52 4},
  pages={
          R2493-R2496
        },
  url={https://api.semanticscholar.org/CorpusID:30510079}
}

@article{Dougherty2008,
  title={Secret-sharing schemes based on self-dual codes},
  author={Steven T. Dougherty and Sihem Mesnager and Patrick Sol{\'e}},
  journal={2008 IEEE Information Theory Workshop},
  year={2008},
  pages={338-342},
  url={https://api.semanticscholar.org/CorpusID:16950187}
}

@article{Cramer2005,
  title={On Codes, Matroids, and Secure Multiparty Computation From Linear Secret-Sharing Schemes},
  author={Ronald Cramer and Vanesa Daza and Ignacio Gracia and Jorge Jim{\'e}nez Urr{\'o}z and Gregor Leander and Jaume Mart{\'i}-Farr{\'e} and Carles Padr{\'o}},
  journal={IEEE Transactions on Information Theory},
  year={2005},
  volume={54},
  pages={2644-2657},
  url={https://api.semanticscholar.org/CorpusID:478302}
}

@book {Conway1999,
    AUTHOR = {Conway, J. H. and Sloane, N. J. A.},
     TITLE = {Sphere packings, lattices and groups},
    SERIES = {Grundlehren der mathematischen Wissenschaften}, 
    VOLUME = {290},
   EDITION = {Third},
 PUBLISHER = {Springer-Verlag, New York},
      YEAR = {1999},
     PAGES = {lxxiv+703},
      ISBN = {0-387-98585-9},
   MRCLASS = {11H31 (05B40 11H06 20D08 52C07 52C17 94B75 94C30)},
  MRNUMBER = {1662447},
MRREVIEWER = {Renaud\ Coulangeon},
       DOI = {10.1007/978-1-4757-6568-7},
       URL = {https://doi-org.uaeu.idm.oclc.org/10.1007/978-1-4757-6568-7},
}

@article{Liu2025,
  title={Column Twisted {R}eed-{S}olomon Codes as {MDS} Codes},
  author={Wei Liu and Jinquan Luo and Puyin Wang and Dengxin Zhai},
  journal={ArXiv:2507.08755},
  year={2025},
  volume={},
  url={https://api.semanticscholar.org/CorpusID:280292300}
}

@article{Sakakibara2013,
  title={Application of random relaying of partitioned {MDS} codeword block to persistent relay {CSMA} over random error channels},
  author={Katsumi Sakakibara and Jumpei Taketsugu},
  journal={2013 5th International Congress on Ultra Modern Telecommunications and Control Systems and Workshops (ICUMT)},
  year={2013},
  pages={106-112},
  url={https://api.semanticscholar.org/CorpusID:9432707}
}

@article{Cadambe2011,
  title={Permutation code: Optimal exact-repair of a single failed node in {MDS} code based distributed storage systems},
  author={Viveck R. Cadambe and Cheng Huang and Jin Li},
  journal={2011 IEEE International Symposium on Information Theory Proceedings},
  year={2011},
  pages={1225-1229},
  url={https://api.semanticscholar.org/CorpusID:10710487}
}

@article{Mirandola2015,
  title={Critical Pairs for the Product Singleton Bound},
  author={Diego Mirandola and Gilles Z{\'e}mor},
  journal={IEEE Transactions on Information Theory},
  year={2015},
  volume={61},
  pages={4928-4937},
  url={https://api.semanticscholar.org/CorpusID:3053379}
}

@article{Han2024,
  title={Explicit constructions of {NMDS} self-dual codes},
  author={Dongchun Han and Hanbin Zhang},
  journal={Des. Codes Cryptogr.},
  year={2024},
  volume={92},
  pages={3573-3585},
  url={https://api.semanticscholar.org/CorpusID:271132365}
}

@article{Abdukhalikov2025nonrs,
  title={Some constructions of non-generalized {R}eed-{S}olomon {MDS} Codes},
  author={Kanat Abdukhalikov and Cunsheng Ding and Gyanendra Kumar Verma},
  journal={ArXiv:2506.04080},
  year={2025},
  volume={},
  url={https://api.semanticscholar.org/CorpusID:279154966}
}

@article{Liu2024,
  title={Constructions of non-Generalized {R}eed-{S}olomon {MDS} codes},
  author={Shengwei Liu and Hongwei Liu and Fr{\'e}d{\'e}rique E. Oggier},
  journal={arXiv:2412.08391},
  year={2024},
  volume={},
  url={https://api.semanticscholar.org/CorpusID:274638295}
}

@article{Beelen2022,
  author    = {P. Beelen and S. Puchinger and J. Rosenkilde},
  title     = {Twisted {R}eed--{S}olomon Codes},
  journal   = {IEEE Transactions on Information Theory},
  volume    = {68},
  number    = {5},
  pages     = {3047--3061},
  year      = {2022},
  month     = may,
  doi       = {10.1109/TIT.2022.3149513}
}

@article{Ball2019,
  title     = {Arcs in Finite Projective Spaces},
  author    = {Ball, Simeon and Lavrauw, Michel},
  journal   = {EMS Surveys in Mathematical Sciences},
  volume    = {6},
  number    = {1-2},
  pages     = {133--172},
  year      = {2019},
  doi       = {10.4171/EMSS/33},
  url       = {https://ems.press/journals/emss/articles/16705}
}

@book{Macdonald1995,
  title     = {Symmetric Functions and Hall Polynomials},
  author    = {Macdonald, Ian G.},
  edition   = {2nd},
  year      = {1995},
  publisher = {Oxford University Press},
  address   = {Oxford},
  isbn      = {9780198504504}
}

@article{Ball2012,
  author    = {Simeon Ball},
  title     = {On large subsets of a finite vector space in which every subset of basis size is a basis},
  journal   = {Journal of the European Mathematical Society},
  volume    = {14},
  number    = {3},
  pages     = {733--748},
  year      = {2012},
  doi       = {10.4171/JEMS/322}
}

@article{Segre1955,
  author    = {Beniamino Segre},
  title     = {Curve razionali normali e $k$-archi negli spazi finiti},
  journal   = {Annali di Matematica Pura ed Applicata},
  volume    = {39},
  pages     = {357--378},
  year      = {1955},
  month     = {December},
  doi       = {10.1007/BF02412874}
}

@book{Macwilliams1977,
  author    = {F. J. MacWilliams and N. J. A. Sloane},
  title     = {The Theory of Error-Correcting Codes},
  publisher = {Elsevier},
  address   = {Amsterdam, The Netherlands},
  year      = {1977}
}

@book {Hirschfeld1998,
    AUTHOR = {Hirschfeld, J. W. P.},
     TITLE = {Projective geometries over finite fields},
    SERIES = {Oxford Mathematical Monographs},
   EDITION = {Second},
 PUBLISHER = {The Clarendon Press, Oxford University Press, New York},
      YEAR = {1998},
     PAGES = {xiv+555},
      ISBN = {0-19-850295-8},
   MRCLASS = {51E15 (05B25 51A30)},
  MRNUMBER = {1612570},
MRREVIEWER = {T.\ G.\ Ostrom},
}

@article{Roth1989,
  title={A construction of non-{R}eed-{S}olomon type {MDS} codes},
  author={Ron M. Roth and Abraham Lempel},
  journal={IEEE Transactions on Information Theory},
  year={1989},
  volume={35},
  pages={655-657},
  url={https://api.semanticscholar.org/CorpusID:9352887}
}

@article{Dau2013,
  title={Balanced Sparsest generator matrices for {MDS} codes},
  author={Son Hoang Dau and Wentu Song and Zheng Dong and Chau Yuen},
  journal={2013 IEEE International Symposium on Information Theory},
  year={2013},
  pages={1889-1893},
  url={https://api.semanticscholar.org/CorpusID:6124009}
}

@article{Kokkala2014,
  title={On the Classification of {MDS} Codes},
  author={Janne I. Kokkala and Denis S. Krotov and Patric R. J. {\"O}sterg{\aa}rd},
  journal={IEEE Transactions on Information Theory},
  year={2014},
  volume={61},
  pages={6485-6492},
  url={https://api.semanticscholar.org/CorpusID:15649803}
}

@article{Dau2014,
  title={On the existence of {MDS} codes over small fields with constrained generator matrices},
  author={Son Hoang Dau and Wentu Song and Chau Yuen},
  journal={2014 IEEE International Symposium on Information Theory},
  year={2014},
  pages={1787-1791},
  url={https://api.semanticscholar.org/CorpusID:2210542}
}

@article{Marchi2001,
  title={Polynomials arising in factoring generalized Vandermonde determinants: an algorithm for computing their coefficients},
  author={Stefano De Marchi},
  journal={Mathematical and Computer Modelling},
  year={2001},
  volume={34},
  pages={271-281},
  url={https://api.semanticscholar.org/CorpusID:122701601}
}

@article{Zhi2025,
  title={New {MDS} codes of non-{GRS} type and {NMDS} codes},
  author={Yujie Zhi and Shixin Zhu},
  journal={Discrete Mathematics},
  year={2025},
  volume={348},
  pages={114436},
  url={https://api.semanticscholar.org/CorpusID:276369353}
}

@article{Li2024,
  title={A family of linear codes that are either non-{GRS} {MDS} codes or {NMDS} codes},
  author={Yang Li and Zhonghua Sun and Shixin Zhu},
  journal={IEEE Transactions on Communications},
  year={2025},
  url={https://api.semanticscholar.org/CorpusID:281297288}
}

@article{Wu2024,
  title={More {MDS} codes of non-{R}eed-{S}olomon type},
  author={Yansheng Wu and Ziling Heng and Chengju Li and Cunsheng Ding},
  journal={arXiv:2401.03391},
  year={2024},
  volume={},
  url={https://api.semanticscholar.org/CorpusID:266844864}
}

@article{Lavauzelle2019,
  title={Cryptanalysis of a system based on twisted {R}eed–{S}olomon codes},
  author={Julien Lavauzelle and Julian Renner},
  journal={Designs, Codes and Cryptography},
  year={2019},
  volume={88},
  pages={1285 - 1300},
  url={https://api.semanticscholar.org/CorpusID:135463487}
}

@INPROCEEDINGS{beelen2018,
  author={Beelen, Peter and Bossert, Martin and Puchinger, Sven and Rosenkilde, Johan},
  booktitle={2018 IEEE International Symposium on Information Theory (ISIT)}, 
  title={Structural Properties of Twisted {R}eed-{S}olomon Codes with Applications to Cryptography}, 
  year={2018},
  volume={},
  number={},
  pages={946-950},
  doi={10.1109/ISIT.2018.8437923}}

@article{Jin2024,
  title={New families of non-{R}eed-{S}olomon {MDS} codes},
  author={Lingfei Jin and Liming Ma and Chaoping Xing and Haiyan Zhou},
  journal={arXiv:2411.14779},
  year={2024},
  volume={},
  url={https://api.semanticscholar.org/CorpusID:274192322}
}

\end{document}